\documentclass[12pt]{article}

\usepackage{times}

\usepackage{xcolor}
\usepackage{graphicx}
\usepackage{epstopdf}
\usepackage{hyperref}
\hypersetup{pdfstartview=} 
\usepackage{url}
\usepackage{amssymb}
\usepackage{amsmath}
\usepackage{amsfonts}
\usepackage{amsthm}
\usepackage[numbers]{natbib}
\usepackage[normalem]{ulem}
\usepackage{stmaryrd}

\newtheorem*{Theorem}{Theorem}
\newtheorem*{Lemma}{Lemma}

\renewcommand{\eqref}[1]{Eq. \ref{#1}}

\newif\ifcmnt
\cmnttrue
\ifdefined\cmntsoff\cmntfalse\fi

\ifcmnt
    \providecommand{\aucmnt}[1]{#1}
\else
    \providecommand{\aucmnt}[1]{}
\fi

\newcommand{\ignore}[1]{}

\newcommand{\bfABXY}{\mathbf{ABXY}}
\newcommand{\Vthresh}{v_{\mathrm{thresh}}}

\newenvironment{sciabstract}{%
\begin{quote} \bf}
{\end{quote}}

\newcounter{lastnote}

\title{Experimentally Generated Random Numbers Certified by the Impossibility of Superluminal Signaling}

\author
{Peter Bierhorst,$^{1\ast}$ Emanuel Knill,$^{1,6}$  Scott Glancy,$^{1}$\\
 Alan Mink,$^{2,3}$ Stephen Jordan,$^{2}$ Andrea Rommal,$^{4}$ Yi-Kai Liu,$^{2}$ \\
 Bradley Christensen,$^{5}$ Sae Woo Nam,$^{1}$  Lynden K. Shalm$^{1}$\\
\\
\normalsize{$^{1}$National Institute of Standards and Technology,}
\normalsize{Boulder 80305, CO, USA}\\
\normalsize{$^{2}$ National Institute of Standards and Technology,}
\normalsize{Gaithersburg 20899, MD, USA}\\
\normalsize{$^{3}$ Theiss Research,}
\normalsize{La Jolla, CA, 92037, USA}\\
\normalsize{$^{4}$ Muhlenberg College, Allentown, PA, 18104, USA}\\
\normalsize{$^{5}$ Department of Physics, University of Wisconsin, Madison, WI,
53706, USA}\\
\normalsize{$^{6}$Center for Theory of Quantum Matter, University of Colorado, Boulder, Colorado 80309, USA}\\
\\
\normalsize{$^\ast$ E-mail:  peter.bierhorst@nist.gov}
}

\date{}

\begin{document} 



\maketitle


\begin{sciabstract}
  Random numbers are an important resource for applications such as
  numerical simulation and secure communication. However, it is
  difficult to certify whether a physical random number generator is
  truly unpredictable. Here, we exploit the phenomenon of quantum
  nonlocality in a loophole-free photonic Bell test experiment for the
  generation of randomness that cannot be predicted within any
  physical theory that allows one to make independent measurement
  choices and prohibits superluminal signaling. To certify and
  quantify the randomness, we describe a new protocol that performs
  well in an experimental regime characterized by low violation of
  Bell inequalities. Applying an extractor function to our data, we
  obtained $256$ new random bits, uniform to within $0.001$.
\end{sciabstract}





Random numbers have many uses. A motivating application for our
experiment is a public randomness beacon that broadcasts certified
random bits at predetermined times \cite{fischer:2011}. For certain
applications, such as sampling and numerical simulation,
algorithmically generated pseudorandom strings are often
sufficient. However, the predictability of pseudorandom strings makes
them unsuitable for other applications such as secure
communication. For such purposes, the theoretical unpredictability of
quantum mechanical experiments make them good candidates for random
number generation \cite{bera:2016}.

A simple quantum random number generator may consist of a device that
measures a pure state of a two-level system in a basis that is not
aligned with the input state. To assert the presence of randomness in
the output of such a device, one must assume that the state and
measurement are properly characterized.  However, this assumption can
be compromised in a potentially undetectable manner. For example, if
predictable sources of noise infect the device, the output may become
increasingly predictable without inducing the failure of statistical
tests of randomness \cite{pironio:2013, miller:2014}. This issue has
inspired the development of the field of device-independent quantum
randomness generation in recent years \cite{pironio:2013, miller:2014,
colbeck:2011, pironio:2010, vazirani:2012, fehr:2013, chung:2014, nieto:2014,
   bancal:2014, thinh:2016}. In the device-independent
paradigm, randomness is generated through an experiment called a Bell
test \cite{BELL}. In its simplest form, a Bell test performs
measurements on an entangled system located in two physically
separated measurement stations, where at each station there are two
types of measurements that can be made. After multiple experimental
trials with varying measurement choices, if the measurement data
violates conditions known as ``Bell inequalities,'' then the data can
be certified to contain randomness under very weak assumptions.

Here, we report the generation of 256 new random bits uniform to
within 0.001 with a ``loophole-free'' Bell test, which notably is
characterized by high detection efficiency and space-like separation
of the measurement stations during each experimental trial.  The bits
are unpredictable assuming that (1) the choices of measurement
settings are independent of the experimental devices and pre-existing
classical information about them and (2) in each experimental trial,
the measurement outcomes at each station are independent of the
settings choices at the other station.  The first assumption is
ultimately untestable, but the premise that it is possible to choose
measurement settings independently of a system being measured is often
tacitly invoked in the interpretation of many scientific experiments
and laws of physics \cite{bell:1985}.  The second assumption can only
be violated if one admits a theory that permits sending signals faster
than the speed of light, given space-like separation of the stations.
We trust the recording and timing electronics to accurately verify
the space-like separation of the relevant events in the experiment,
and that the classical computing equipment used to process the data
operates according to specification.  Under these assumptions, the
output randomness is certified to be unpredictable with respect to a
real or hypothetical actor ``Eve'' in possession of the pre-existing
classical information and physically isolated from the devices while
they are under our control. The bits remain unpredictable to Eve if
she learns the settings at any time after her last interaction with
the devices.  If the devices are trusted, which is reasonable if we
built them, it may be the case that Eve has no pre-existing
information. The settings can then come from public randomness
generated externally at any time \cite{pironio:2013}.  Our framework
encompasses any quantum description of the system being measured, as
well as more general theoretical possibilities \cite{PRBOX}.

The only previous experimental production of certified randomness from
Bell test data was reported in the ground-breaking paper by Pironio et
al. \cite{pironio:2010}.  Their Bell test was implemented with ions in
two separate ion-traps, closing the detection loophole
\cite{pearle:1970} but without space-like separation. Indeed, Bell
tests achieving space-like separation without leaving other
experimental loopholes open have been performed only recently
\cite{hensen:2015, shalm:2015, giustina:2015,rosenfeld:2016}. Under
more restrictive assumptions than ours, the maximum amount of
randomness in principle available in the data of Pironio et al. was
quantified as $42$ bits with an error parameter of $0.01$.  However,
they did not extract a uniformly distributed bit string from their
data set. 

We generated randomness using a photonic loophole-free Bell test,
illustrated in Fig. \ref{expschematic}.  The experiment consisted of a
source of entangled photons and two measurement stations named
``Alice'' and ``Bob''.  During an experimental trial, at each station
a random choice was made between two measurement settings labeled 0
and 1, after which a measurement outcome of detection (+) or
nondetection (0) was recorded. Each station's choice was space-like
separated from the other station's measurement event. For trial $i$,
we model Alice's settings choices with the random variable $X_i$ and
Bob's with $Y_i$, both of which take values in the set $\{0,1\}$.
Alice's and Bob's measurement outcome random variables are
respectively $A_i$ and $B_i$, both of which take values in the set
$\{\text{+},0\}$. When referring to a generic single trial, we omit
indices. With this notation, a general Bell inequality for our
scenario can be expressed in the form \cite{BBP}
\begin{equation}\label{e:genbellineq}
\sum_{abxy}s^{ab}_{xy}\mathbb{P}(A=a,B=b|X=x,Y=y) \le \beta,
\end{equation}
where the $s^{ab}_{xy}$ are fixed real coefficients indexed by
$a,b,x,y$ that range over all possible values of $A,B,X,Y$.  The upper
bound $\beta$ is required to be satisfied whenever the
settings-conditional outcome probabilities are induced by a model
satisfying ``local realism'' (LR). LR distributions, which cannot be
certified to contain randomness, are those for which
$\mathbb{P}(A=a,B=b|X=x,Y=y)$ is of
the form $\sum_\lambda \mathbb{P}(A=a|X=x,
\Lambda=\lambda)\mathbb{P}(B=b|Y=y,
\Lambda=\lambda)\mathbb{P}(\Lambda=\lambda)$ for a random variable
$\Lambda$ representing local hidden variables. The Bell inequality is
non-trivial if there exists a quantum-realizable distribution that can
violate the bound $\beta$.

\begin{figure}
\begin{center}
\includegraphics{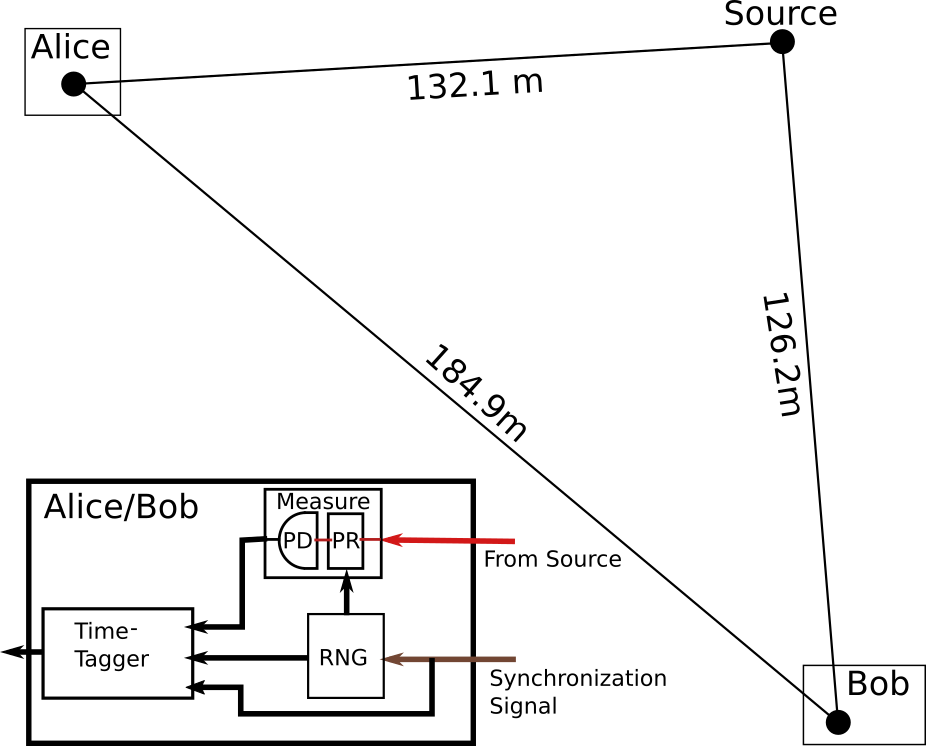}
\end{center}
\caption{Schematic layout of the loophole-free Bell test
  \cite{shalm:2015}. The setup consists of a source located at the
  corner of an ``L'' and two measurement stations, Alice and Bob, at
  the ends of the ``L''. The source generates entangled pairs of
  photons that travel through fiber optic cables to the respective
  measurement stations. At each station, schematically shown in the
  box labeled Alice/Bob, a random number generator (RNG) governs
  polarization rotators (PR) that implement one of two possible
  measurement settings prior to detection or nondetection of an
  arriving photon by a photon detector (PD), achieving a system
  efficiency of about $75 \,\%$. Each station's measurement setting,
  outcome, and synchronization signal are then recorded by its
  timetagger. See Ref.~\cite{shalm:2015} for details.}
\label{expschematic}
\end{figure}


Experimental violations of \eqref{e:genbellineq} indicate the presence
of randomness in the data \cite{bera:2016}. To quantify randomness
with respect to Eve, we represent Eve's initial classical information
by a random variable $E$.  We formalize the assumption that
measurement settings can be generated independently of the system
being measured and Eve's information with the following condition:
\begin{equation}\label{e:mtunifsettings}
\mathbb{P}(X_i=x,Y_i=y|E=e,\text{past}_i)=\mathbb{P}(X_i=x,Y_i=y)=\frac{1}{4} \quad \forall x,y,e,
\end{equation}
where $\text{past}_i$ represents events in the past of the $i$'th
trial, specifically including the trial settings and outcomes for
trial 1 through $i-1$. Our other assumption, that measurement outcomes
are independent of remote measurement choices, is formalized as
follows:
\begin{eqnarray}\label{e:mtnosig}
      \mathbb{P}(A_i=a|X_i=x,Y_i=y,E=e, \text{past}_i) &=& \mathbb{P}(A_i=a|X_i=x,E=e, \text{past}_i)\notag\\
      \mathbb{P}(B_i=b|X_i=x,Y_i=y,E=e, \text{past}_i) &=& \mathbb{P}(B_i=b|Y_i=y,E=e, \text{past}_i)\quad \forall x,y,e.
\end{eqnarray}
These equations are commonly referred to as the ``non-signaling''
assumptions, although they are often stated without the conditionals
$E$ and $\text{past}_i$.  Our space-like separation of choices and
remote measurements provide assurance that the experiment obeys
Eqs.~\ref{e:mtnosig}.  We make no other assumptions on the physics of
the devices used in the experiment, but remark that if one constrains
the devices to quantum physics, constraints stronger than
non-signaling are possible \cite{navascues:2008}.

Given Eqs.~\ref{e:mtunifsettings} and \ref{e:mtnosig}, our protocol
produces random bits in two sequential parts. For the first part,
``entropy production'', we implement $n$ trials of the Bell test, from
which we compute a statistic $V$ related to a Bell inequality
(\eqref{e:genbellineq}). $V$ quantifies the Bell violation and
determines whether or not the protocol passes or aborts.  If the
protocol passes, we can certify an amount of randomness in the outcome
string even conditioned on the setting string and $E$.  In the second
part, ``extraction,'' we process the outcome string into a shorter
string of bits whose distribution is close to uniform. We used our
customized implementation of the Trevisan extractor
\cite{trevisan:2001} derived from the framework of Mauerer, Portmann
and Scholz \cite{mauerer:2012} and the associated open source code. We
call this the TMPS algorithm, see Supplementary Text (ST)
\ref{st:trevisan} for details.

We developed a new method for the entropy production part of our
protocol, as previous methods
\cite{pironio:2013,miller:2014,colbeck:2011,pironio:2010,vazirani:2012,fehr:2013,chung:2014,coudron:2014,arnon:2016}
are ineffective in our experimental regime (ST \ref{st:previous}),
which is characterized by a small per-trial violation of Bell
inequalities. Recent papers that explore how to effectively certify
randomness from a wider range of experimental regimes assume that
measured states are independent and identically distributed (i.i.d.)
or that the regime is asymptotic
\cite{nieto:2014,bancal:2014,thinh:2016, miller2:2014}.  Our method,
which does not require these assumptions, builds on the
Prediction-Based Ratio (PBR) method for rejecting LR
\cite{zhang:2011}.  Applying this method to training data (see below),
we obtain a real-valued ``Bell function'' $T$ with arguments $A,B,X,Y$
that satisfies $T(A,B,X,Y) > 0$ with expectation $\mathbb{E}(T) \leq 1$ for any
LR distribution satisfying Eq.~\ref{e:mtunifsettings}. From $T$ we
determine the maximum value $1+m$ of $\mathbb{E}(T)$ over all distributions
satisfying Eqs.~\ref{e:mtunifsettings} and~\ref{e:mtnosig}, where we
require that $m>0$.  Such a function $T$ induces a Bell inequality
(\eqref{e:genbellineq}) with $\beta=4$ and $s^{ab}_{xy}= T(a,b,x,y)$.
Define $T_i=T(A_i,B_i,X_i,Y_i)$ and $V=\prod_{i=1}^nT_i $. If the experimenter observes a value of $V$
larger than $1$, this indicates a violation of the Bell inequality and
the presence of randomness in the data.  The randomness is quantified
by the following theorem, proven in the ST \ref{st:ept}. Below, we
denote all of the settings of both stations with
${\bf XY} = X_1Y_1X_2Y_2...X_nY_n$, and other sequences such as
$\mathbf{AB}$ and $\mathbf{ABXY}$ are similarly interleaved over $n$
trials.

\medskip
\medskip

\noindent{\it Entropy Production Theorem.}  Suppose $T$ is a Bell function
satisfying the above conditions. Then in an experiment of $n$ trials
obeying Eqs.~\ref{e:mtunifsettings} and~\ref{e:mtnosig}, the following
inequality holds for all $\epsilon_{\text{p}} \in (0,1)$ and
$\Vthresh$ satisfying $1\le \Vthresh \le (1+(3/2)m)^{n}\epsilon_{\mathrm{p}}^{-1}$:
\begin{equation}\label{e:theorem}
\mathbb{P}_e\left(\mathbb{P}_e({\bf AB}|{\bf XY})> \delta \text{ AND } V\ge \Vthresh \right)\le \epsilon_{p}
\end{equation}
where $\delta = [1+(1-\sqrt[n]{\epsilon_{\text{p}}\Vthresh})/(2m)]^n$
and $\mathbb{P}_e$ denotes the probability distribution conditioned on the
event $\{E=e\}$, where $e$ is arbitrary. The expression
$\mathbb{P}_e({\bf AB}|{\bf XY})$ denotes the random variable that takes the
value $\mathbb{P}_{e}(\mathbf{AB}=\mathbf{ab}|\mathbf{XY}=\mathbf{xy})$ when
$\mathbf{ABXY}$ takes the value $\mathbf{abxy}$.

\medskip
\medskip

\noindent In words, the theorem says that with high probability, if
$V$ is at least as large as $\Vthresh$, then the output $\mathbf{AB}$
is unpredictable, in the sense that no individual outcome
$\{\mathbf{AB}=\mathbf{ab}\}$ occurs with probability higher than
$\delta$, even given the information
$\{\mathbf{XY}E=\mathbf{xy}e\}$. The theorem supports a protocol that
aborts if $V$ takes a value less than $\Vthresh$, and passes
otherwise. In order for the guarantee of the theorem to be meaningful,
it is necessary to have a lower bound $\kappa$ on
$\mathbb{P}(\text{pass}) = \mathbb{P}(V\geq \Vthresh)$. While we
cannot be sure of such a lower bound, an assumption that the
probability of passing exceeds some positive value is necessary,
because for any implementation of the protocol there is always a
completely predictable LR theory with positive passing probability,
however small.  If $\kappa$ were $1$, then $-\log_2(\delta)$ would be
a so-called ``smoothed min-entropy'', a quantity that characterizes
the number of uniform bits of randomness that are in principle
available in $\mathbf{AB}$ \cite{trevisan:2000, renner:2006}.  We show
in the ST~\ref{st:T} that for constant $\epsilon_{\text{p}}$ and
$\kappa$, $-\log_{2}(\delta)$ is proportional to the number of trials.
How many bits we can actually extract depends on $\kappa$ and
$\epsilon_{\text{fin}}$, the output's maximum allowed distance from
uniform.

To extract the available randomness in ${\bf AB}$, we use the TMPS
algorithm to obtain an extractor, specifically a function $\text{Ext}$
that takes as input the string ${\bf AB}$ and a length $d$ ``seed''
bit string ${\bf S}$, where ${\bf S}$ is uniform and independent of
${\bf ABXY}$. Its output is a length $t$ bit string. Note that ${\bf
  S}$ can be obtained from $d$ additional instances of the random
variables $X_i$, so \eqref{e:mtunifsettings} ensures the needed
independence and uniformity conditions on ${\bf S}$.  In order for the
output to be within a distance $\epsilon_{\text{fin}}$ of uniform
independent of ${\bf XY}$ and $E$, the entropy production and
extractor parameters must satisfy the constraints given in the next
theorem, proven in the ST \ref{st:pst}.  The measure of distance used
is the ``total variation distance'', expressed by the left-hand side
of Eq.~\ref{e:mtfinal} in the theorem.

\medskip\medskip

\noindent{\it Protocol Soundness Theorem.}
Let $0<\epsilon_{\text{ext}}<1$. Suppose that the protocol parameters satisfy 
\begin{equation}\label{e:mttrev1}
t+4\log_2t \le -\log_2 \delta + \log_2 \kappa +5\log_2 \epsilon_{\text{ext}} -11
\end{equation}
with $\kappa\leq \mathbb{P}(\text{pass})$. Then the
function $\text{Ext}$ obtained by the TMPS algorithm satisfies
\begin{multline}\label{e:mtfinal}
  \frac{1}{2}\sum_{z,{\bf xyes}} \Big|\mathbb{P}\big(\text{Ext}({\bf
    AB},{\bf S})=z,{\bf XYES}={\bf
    xyes}|\text{pass}\big)-\mathbb{P}_{\text{unif}}(Z=z)\mathbb{P}\big({\bf
    XYE=xye}|\text{pass}\big)\mathbb{P}_{\text{unif}}({\bf S}={\bf s})\Big|\\ \le
  \epsilon_{\text{fin}}=\epsilon_{\text{p}}/\kappa+\epsilon_{\text{ext}},
\end{multline}
where $\mathbb{P}_{\text{unif}}$ denotes the uniform probability distribution.

\medskip
\medskip

\noindent The number of seed bits $d$ required satisfies
$d = O(\log(t)\log(nt/\epsilon)^{2})$. An explicit upper bound on
$d$ for our extractor implementation is given in the ST
\ref{st:trevisan}.


As the primary demonstration of our protocol, we applied it to the
main data set analyzed in \cite{shalm:2015}, which is titled ``XOR 3''
and consists of a total of $182,161,215$ trials, acquired in
$30\,\text{min}$ of running the experiment, improving on the
approximately one month duration of data acquisition reported in
Ref.~\cite{pironio:2010}. Before starting the protocol, we set aside
the first $5\times 10^{7}$ trials as training data, which we used to
choose parameters needed by the protocol.  With the training data
removed, the number $n$ of trials used by the protocol was
$132,161,215$. We used the training data to determine a Bell function
$T$ with statistically strong violation of LR on the training data
according to the PBR method \cite{zhang:2011}; see ST \ref{st:T}. The
function $T$ obtained is given in Table \ref{t:PBR}. A sample of $n$
i.i.d. trials from the distribution inferred from the training data
would have an approximate $0.95$ probability for $V$
to exceed $1.66\times 10^6$.  Based on this estimate, we chose
$\Vthresh=1.66\times 10^{6}$ and then numerically studied the
constraints on the number $t$ of bits extracted and the final error
$\epsilon_{\text{fin}}$ as a function of $\epsilon_{\text{p}}$,
$\epsilon_{\text{ext}}$, and $\kappa$. Based on this study we chose a
benchmark goal of $t=256$ and $\epsilon_{\text{fin}}=0.001$, with the
constraints satisfied for $\epsilon_{\text{p}}=3.1797 \times 10^{-4}$,
$\epsilon_{\text{ext}}=3.533 \times 10^{-5}$, and $\kappa=0.33$. We
chose $\epsilon_{\text{p}}$ and $\epsilon_{\text{ext}}$ in the ratio
$9{:}1$, which was generally found to be near-optimal when numerically
maximizing $t$ in \eqref{e:mttrev1} for fixed values of
$\epsilon_{\text{fin}}$. We chose $\kappa$ to be safely below $0.95$
(the estimate of the probability to exceed $\Vthresh$ based on the
training data) so that the protocol would be more robust against
drifts in trial statistics or the possibility of mid-experiment
equipment malfunction. For commissioning purposes, we examined how the
protocol behaves in six earlier runs of the experiment
\cite{shalm:2015}. Including XOR 3 and the two blind data sets
described below, the value of $\Vthresh$ chosen in the same manner is
exceeded seven times, suggesting that our choice of $\kappa$ is
reasonable. 

\begin{table}
  \centering\caption{The Bell function $T$ obtained from the training
    data. We used a maximum likelihood method to infer a non-signaling
    distribution based on the raw counts of the training trials,
    namely the first $5\times 10^{7}$ trials of data set XOR 3. We then
    determined the function $T$ that maximizes $\mathbb{E}(\ln T)$
    according to this distribution, subject to the constraints that
    $\mathbb{E}(T)_{LR}\le 1$ for all LR distributions and
    $T(0,0,x,y)=1$ for all $x,y$. The latter constraint
      improves the signal-to-noise for our data. The function $T$
    yields $m=0.0120275$, and $\mathbb{E}(T) = 1.0000003928$ for the
    non-signaling distribution inferred from the training data. One
    can also interpret the numbers below as the coefficients
    $s^{ab}_{xy}$ in \eqref{e:genbellineq}, which defines a Bell
    inequality with $\beta=4$. The values of $T$ are rounded down at
    the tenth digit.}
\label{t:PBR}
 \begin{tabular}{ r|c|c|c|c| }
 \multicolumn{1}{r}{}
  &  \multicolumn{1}{c}{$ab=\text{++}$}
 &  \multicolumn{1}{c}{$ab=\text{+}0$}
 &  \multicolumn{1}{c}{$ab=0\text{+}$} 
   &  \multicolumn{1}{c}{$ab=00$}
\\
  \cline{2-5}
   $xy=00$&   1.0244479364  &  0.9643897947  &  0.9638375026  & 1 \\
 \cline{2-5}
   $xy=01$&   1.0315040078  &  0.9393895435  &  0.9958939908  & 1 \\
 \cline{2-5}
   $xy=10$&   1.0317342738  &  0.9955719750  &  0.9399418138  & 1 \\
 \cline{2-5}
   $xy=11$&   0.9123069953  &  1.0044279882  &  1.0041059756  & 1 \\
 \cline{2-5}
 \end{tabular}
\end{table} 

Running the protocol on XOR 3 with these parameters, the running
product $\prod_{i=1}^cT_i$ exceeded $\Vthresh$ at trial number
$c=67,173,533$. At this point it is consistent with
Eqs.~\ref{e:mtunifsettings} and \ref{e:mtnosig} to change all outcomes
to $0$ for the remainder of the run, ensuring that the remaining
values of $T_{i}$ are $T_{i}=1$. This is justified because our assumptions allow for Alice
  and Bob to cooperatively make arbitrary changes to the experiment in
  advance of a trial based on the past, which includes the current
  running product. Turning off the detectors to guarantee outcomes of
  $0$ is one such change.  In principle there was sufficient time (at
  least $5\,\mu s$) for the necessary communication to take place
  after the previous trial. We
thereby preserve randomness accumulated at trial $c$ even if
statistical fluctuations or experimental drift would otherwise cause
$V$ to be less than $\prod_{i=1}^cT_i$.  In the ST~\ref{st:ept},
soundness of this procedure is established by proving a
past-parametrized version of the Entropy Production Theorem.  Applying
the extractor to the resulting output string ${\bf AB}$ with a seed of
length $d=73,947$, we extracted $256$ bits, certified to be uniform to
within $0.001$. They are, in hexadecimal form:
\begin{equation*}
\text{D731F577BC44F4993E28A84E44EEBD7824C09D203772F876F67D13D3C974FBC2}
\end{equation*}

The largest running product of the $T_i$'s observed during the
protocol without changing outcomes to $0$ was $2.76 \times 10^9$. If
we had chosen this value for $\Vthresh$, then $256$ bits
could have been extracted to within $0.001$ of uniform for any value
of $\kappa$ exceeding $4.85\times10^{-4}$. Figure \ref{f:XOR3}
displays the nominal smoothed min-entropy
  $-\log_{2}(\delta)$ and extractable bits
for alternative choices of $\epsilon_{\text{fin}}$.

\begin{figure}\centering
\includegraphics{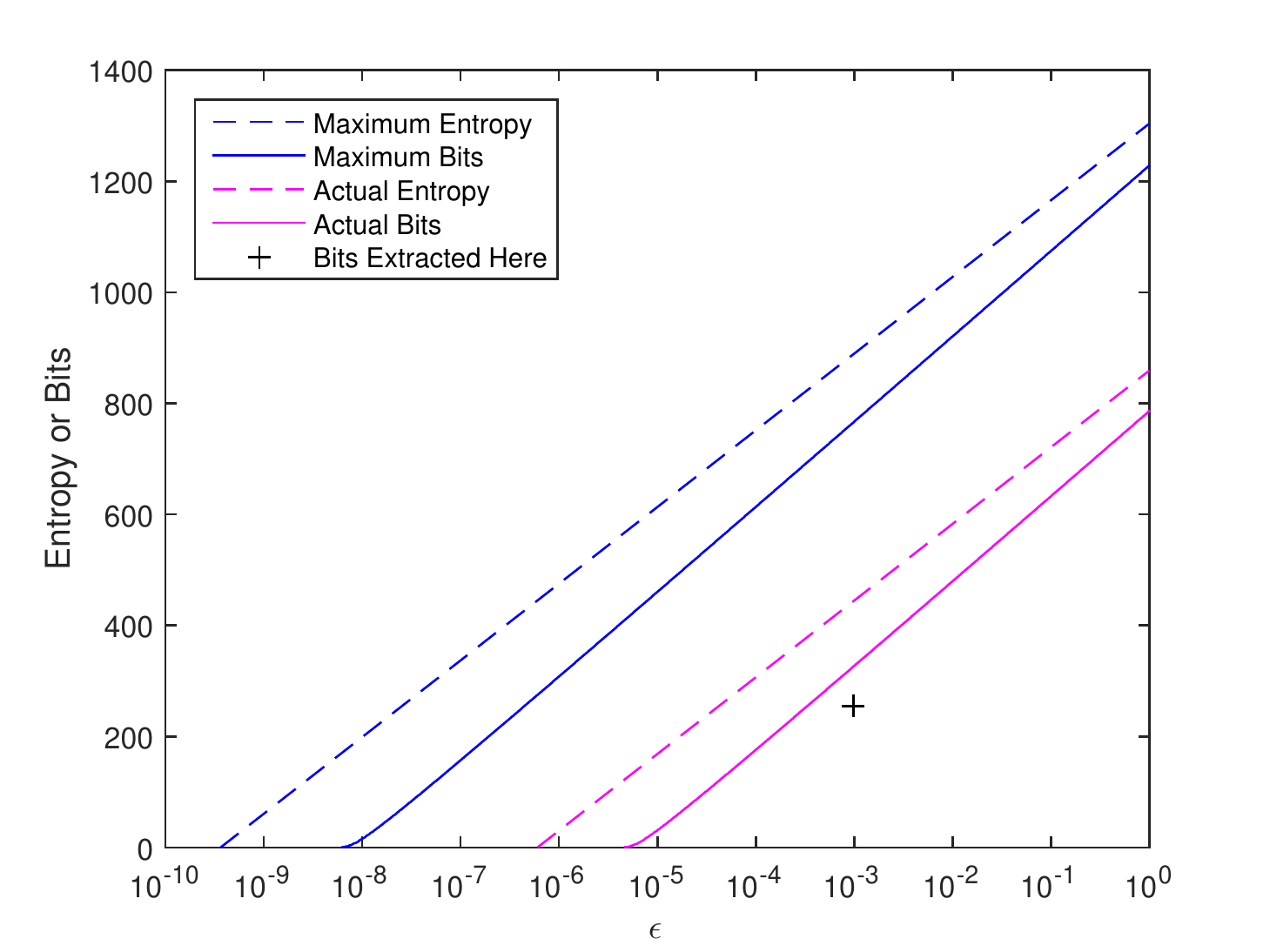}
\caption{ Entropy and extractable bits as a function of error for data
  set XOR 3.  The figure shows the tradeoff between final error and
  number of extractable bits.  The horizontal axis shows the error
  $\epsilon$, which is $\epsilon=\epsilon_{\text{p}}$ for the dashed
  curves and $\epsilon=\epsilon_{\text{fin}}$ for the solid curves.
  The vertical axis shows $-\log_2(\delta)$ (referred to as
  ``entropy'', a bound on the available number of random bits) for the
  dashed curves and $t$ (the number of bits extractable by the TMPS
  algorithm) for the solid curves. The dashed curves are based on the
  Entropy Production Theorem and the solid ones on the Protocol
  Soundness Theorem. For the solid curves, we set $\kappa=1$,
  $\epsilon_{\text{p}}=0.9\,\epsilon_{\text{fin}}$, and
  $\epsilon_{\text{ext}}=0.1\,\epsilon_{\text{fin}}$. The upper curves
  labeled by ``maximum'' show the maximum number of extractable bits,
  obtained if we had set $\Vthresh=2.76\times 10^9$, which is the
  maximum running product of the $T_i$. The lower curves labeled by
  ``actual'' use $\Vthresh=1.66\times 10^6$, which we chose before
  running the protocol. The ``+'' symbol corresponds to
  $\epsilon_{\text{fin}}=0.001$ and 256 extracted bits, and lies off
  the lower solid curve because this value of $\epsilon_{\text{fin}}$
  was calculated with $\kappa=0.33< 1$.}
\label{f:XOR3}
\end{figure}

The data set XOR 3 was previously subject to statistical analysis
\cite{shalm:2015}, and hence the above protocol was run with advance
knowledge of its features, in particular, that it performed well as a
test against LR. However, two ``blind'' data sets, ``Blind 1'' and
``Blind 2'' were recorded as part of the original experiment
\cite{shalm:2015} and saved for future analysis. After the development
of the protocol of this paper, these data sets were analyzed for the
first time according to the new methods, and the results are
summarized in Figure \ref{f:blinds}. Further details on all data
  sets considered, analyses applied and results are in the
  ST \ref{st:actual}.

\begin{figure}\centering
\includegraphics[scale=1]{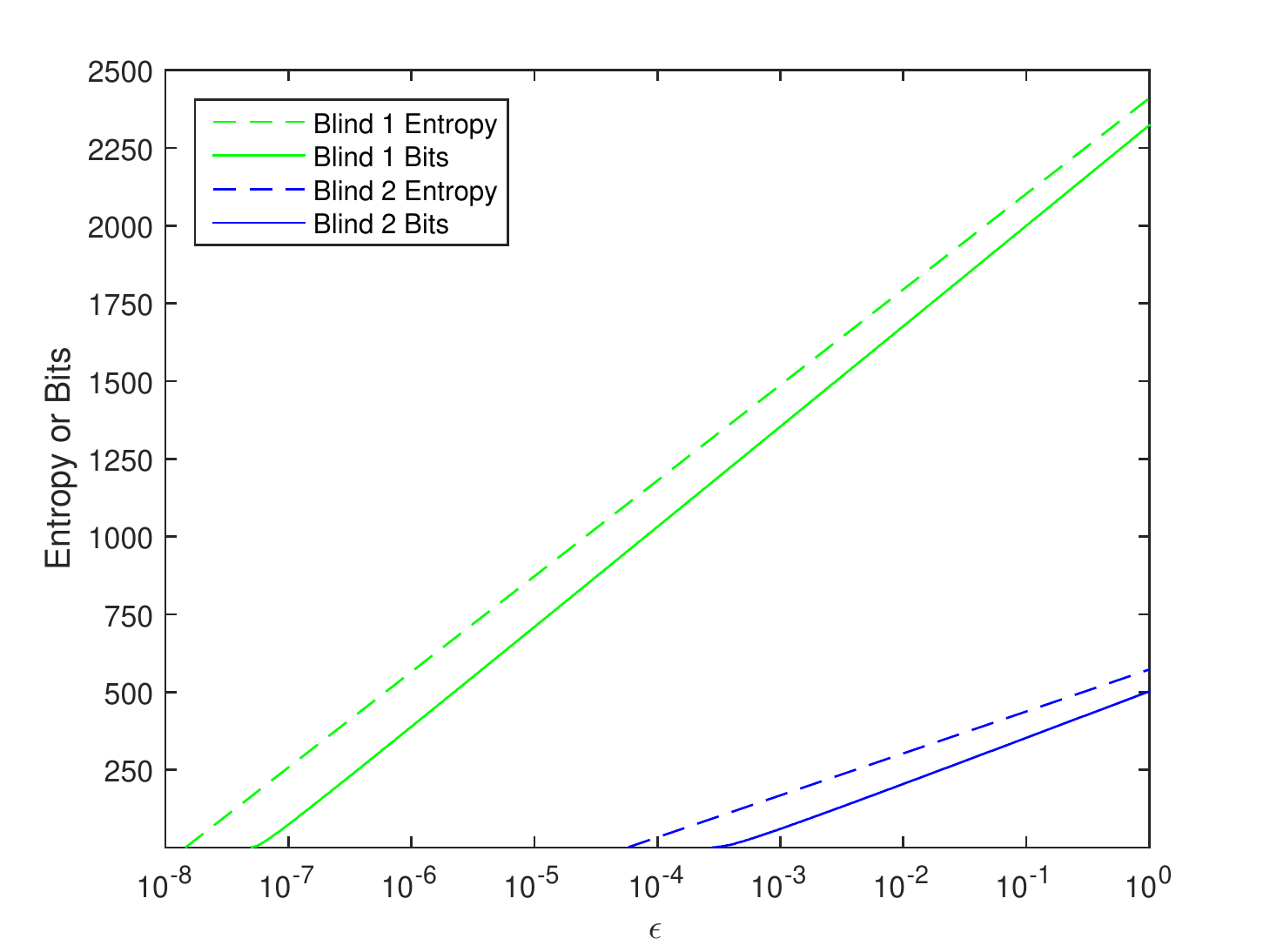}
\caption{Entropy and extractable bits as a function of error for
    Blind 1 and 2.  The figure shows the tradeoff between final error
    and number of extractable bits for the blind data sets. The axes and
    distinction between dashed and solid curves are as in
    Fig.~\ref{f:XOR3}, but only the curves for $\Vthresh$
    given by the maximum observed running product of the $T_i$'s is
    shown. Data sets Blind 1 and Blind 2 have parameters $m=0.00540,
    n=306,464,574, \Vthresh=6.88\times10^7$ and $m=0.01231,
    n=162,837,253, \Vthresh=17,528$, respectively. We did not
    explicitly extract bits for these data sets.}
\label{f:blinds}
\end{figure}

In conclusion, we have demonstrated a protocol that extracts
randomness from experimentally feasible low-violation photonic
experiments. We obtained $256$ new near-uniform random bits from
$1.32\times 10^{8}$ trials. In the process we used $2.64\times 10^{8}$
uniform bits to choose the settings and $7.34\times 10^{4}$ for the
extractor seed. Because the extractor used is a ``strong'' extractor,
the seed bits are still uniform conditional on passing, so they can be
recovered at the end of the protocol for uses elsewhere. However, this
is not the case for the settings-choice bits because the probability
of passing is less than $1$. To reduce the entropy used for the
settings, our protocol can be modified to use highly biased settings
choices \cite{pironio:2010}. It is also helpful to observe that the
output bits are unknown to external observers such as Eve even if the
settings and seed bits are public \cite{pironio:2013,bancal:2014}; the
only requirement is that the experiment's physical isolation ensures
that these public bits are random relative to the experimental devices
before they are used and, if Eve ever had access to the devices, to
Eve at the time of access.  In particular, if there never was a
physical connection to Eve, there is no restriction on settings and
seed bit knowledge.  When applied in this way, the protocol is a
method for private randomness generation.

For future work, we hope to take advantage of the adaptive
capabilities of the Entropy Production Theorem (ST~\ref{st:ept}) to
dynamically compensate for experimental drift during run time, and in
view of advances toward practical quantum computing it is desirable to
study the protocol in the presence of quantum side information. We
also look forward to technical improvements in experimental equipment
for larger violation and higher trial rates. These will enable faster
generation of random bits with lower error and support the use of
biased settings choices.

\medskip

\paragraph{Acknowledgments}
This work is a contribution of the National Institute of Standards and
Technology and is not subject to U.S. copyright. We thank Carl
  Miller and Kevin Coakley for helpful discussions and feedback on the
  manuscript.

\newpage


\begin{center}
{\Large Experimentally Generated Random Numbers Certified}\\
{\Large by the Impossibility of Superluminal Signaling}\\
{\large (Supporting Online Material)}\\
Peter Bierhorst, Emanuel Knill,  Scott Glancy,\\
 Alan Mink, Stephen Jordan, Andrea Rommal, Yi-Kai Liu, \\
 Bradley Christensen, Sae Woo Nam,  Lynden K. Shalm\\
\end{center}

\def\thesection{S} 

After preliminaries to establish notation and summarize needed
properties of total variation distance and non-signaling
  distributions in \ref{st:prelim}, we give the proof of the Entropy
  Production Theorem in \ref{st:ept}. We explain how we chose the
  Bell function $T$, whose product determines whether we obtained
  the desired amount of randomness, in \ref{st:T}. We then discuss the
  parameters of the extractors obtained by the TMPS algorithm
  (\ref{st:trevisan}) and prove the Protocol Soundness Theorem
  (\ref{st:pst}). Details on how we analyzed the experimental data
  sets are in \ref{st:actual}. Justification for our claim that
  previous methods do not obtain any randomness from our weakly
  violating data is given in \ref{st:previous}.

\subsection{Preliminaries}
\label{st:prelim}

We use the standard convention that capital letters refer to random
variables (RVs) and corresponding lowercase letters refer to values
that the RVs can take. All our RVs take values in finite sets such as
the set of bit strings of a given length or a finite subset of the
reals, so that our RVs can be viewed as functions on a finite
probability space. We usually just work with the induced joint
distributions on the sets of values assumed by the RVs.  When working
with conditional probabilities, we implicitly exclude points where the
conditioner has zero probability whenever appropriate.  We use
$\mathbb{P}(\ldots)$ to denote probabilities and $\mathbb{E}(\ldots)$
for expectations.  Inside $\mathbb{P}(\ldots)$ and when used as
conditioners, logical statements involving RVs are event
specifications to be interpreted as the event for which the statement
is true.  For example, $\mathbb{P}(R>\delta)$ is equivalent to
$\mathbb{P}(\{\omega:R(\omega)>\delta\})$, which is the probability of
the event that the RV $R$ takes a value greater than $\delta$. The
same convention applies when denoting events with $\{\ldots\}$. For
example, the event in the previous example is written as
$\{R>\delta\}$. While formally events are sets, we commonly use
logical language to describe relationships between events. For
example, the statement that $\{R>\delta\}$ implies $\{S>\epsilon\}$
means that as a set, $\{R>\delta\}$ is contained in
$\{S>\epsilon\}$. When they appear outside the the
mentioned contexts, logical statements are constraints on RVs. For
example, the statement $R>\delta$ means that all values $r$ of $R$
satisfy $r>\delta$, or equivalently, for all $\omega$,
$R(\omega)>\delta$.  As usual, comma separated statements are combined
conjunctively (with ``and'').  (In the main text, for clarity, we have
used an explicit ``AND'' for this purpose.)

If there are free RVs inside $\mathbb{P}(\ldots)$ or in the
conditioner of $\mathbb{E}(\ldots|\ldots)$ outside event specifications, the
final expression defines a new RV as a function of the free RVs.  An
example from the Entropy Production Theorem is the expression
$\mathbb{P}({\bf AB}|{\bf XY})$, which defines the RV that takes the value
$\mathbb{P}({\bf AB}={\bf ab}|{\bf XY}={\bf xy})$ when the event
$\{{\bf ABXY}={\bf abxy}\}$ occurs.  Values of RVs such as ${\bf x}$
appearing by themselves in $\mathbb{P}(\ldots)$ denote the event
$\{{\bf X}={\bf x}\}$. Thus we abbreviate expressions such as
$\mathbb{P}({\bf AB}={\bf ab}|{\bf XY}={\bf xy})$ by $\mathbb{P}({\bf ab}| {\bf xy})$.
Sometimes it is necessary to disambiguate the probability distribution
with respect to which $\mathbb{E}(\ldots)$ is to be computed.  In such cases we
use a subscript at the end of the expression consisting of a symbol
for the probability distribution, so $\mathbb{E}(T)_\mathbb{Q}$ is the expectation of
$T$ with respect to the distribution $\mathbb{Q}$. In a few instances, we use
$\llbracket\phi\rrbracket$ for logical expressions $\phi$ to denote
the $\{0,1\}$-valued function evaluating to $1$ iff $\phi$ is true.

The amount of randomness that can be extracted from an RV $R$ is
  quantified by the \textit{min-entropy}, defined as $-\log_2 \max_r
  \mathbb{P}(R=r)$. The error of the output of an extractor is given as the
  \textit{total variation} (TV) distance from uniform. Given two
  probability distributions $\mathbb{P}_1$ and $\mathbb{P}_{2}$ for $R$, the TV
  distance between them is given by
\begin{eqnarray}\label{e:condTV}
  \text{TV}(\mathbb{P}_1,\mathbb{P}_2)&=&\frac{1}{2} \sum_r \left| \mathbb{P}_1(R=r)-\mathbb{P}_2(R=r)\right|\notag\\
  &=& \sum_{r:\mathbb{P}_{1}(r)>\mathbb{P}_{2}(r)}\left( \mathbb{P}_{1}(R=r)-\mathbb{P}_{2}(R=r)\right)\notag\\
  &=& \sum_{r}\llbracket \mathbb{P}_{1}(r)>\mathbb{P}_{2}(r)\rrbracket\left( \mathbb{P}_{1}(R=r)-\mathbb{P}_{2}(R=r)\right).
\end{eqnarray}
We sometimes compute TV distances for distributions of specific RVs,
conditional or unconditional ones. For this we introduce the notation
$\mathbb{P}_X$ for the distribution of values of $X$ according to $\mathbb{P}$, and
$\mathbb{P}_{X|Y=y}$ for the distribution of $X$ conditioned on the event
$\{Y=y\}$. With this notation, $\mathbb{P}_X\mathbb{P}_Y$ refers to the product distribution that assigns probability $\mathbb{P}_X(X=x)\mathbb{P}_Y(Y=y)$ to the event $\{X=x,Y=y\}$.

For the proof of the Protocol Soundness Theorem, we need three results
involving the TV distance.  The first is the triangle inequality for
TV distances, see Ref.~\cite{levin:2009} for this and other basic
properties of TV distances.
\begin{equation}\label{e:TIforTV}
\text{TV}(\mathbb{P}_{1},\mathbb{P}_{3}) \leq \text{TV}(\mathbb{P}_{1},\mathbb{P}_{2})+\text{TV}(\mathbb{P}_{2},\mathbb{P}_{3}).
\end{equation}

According to the second result, if $\mathbb{P}$ and $\mathbb{Q}$ are joint distributions
of RVs $V$ and $W$, where the marginals of $W$ satisfy $\mathbb{P}(w)=\mathbb{Q}(w)$,
then the distance between them is given by the average conditional
distance. This is explicitly calculated as follows:
\begin{eqnarray}\label{e:TVsameconditionals}
\text{TV}(\mathbb{P}_{VW},\mathbb{Q}_{VW}) &=&
  \sum_{w}\sum_{v}\llbracket \mathbb{P}(v,w)>\mathbb{Q}(v,w)\rrbracket
    \left(\mathbb{P}(v,w)-\mathbb{Q}(v,w)\right)\notag\\
  &=&
  \sum_{w}\sum_{v}\llbracket \mathbb{P}(v|w)\mathbb{P}(w)>\mathbb{Q}(v|w)\mathbb{Q}(w)\rrbracket
    \left(\mathbb{P}(v|w)\mathbb{P}(w)-\mathbb{Q}(v|w)\mathbb{Q}(w)\right)\notag\\
  &=& 
  \sum_{w}\sum_{v}\llbracket \mathbb{P}(v|w)>\mathbb{Q}(v|w)\rrbracket
    \left(\mathbb{P}(v|w)-\mathbb{Q}(v|w)\right)\mathbb{P}(w)\notag\\
  &=& \sum_{w}\text{TV}(\mathbb{P}_{V|W=w},\mathbb{Q}_{V|W=w})\mathbb{P}(w).
\end{eqnarray}

The third result is a special case of the data-processing inequality
for TV distance. See Ref.~\cite{pardo:1997} for this and many other
data-processing inequalities.  Let $V$ be a random variable taking
values in a finite set $\mathcal V$, and let
$F:\mathcal V \to \mathcal W$ be a function so that $F(V)$ is a random
variable taking values in the set $\mathcal W$. Then if $\mathbb{P}$ and $\mathbb{Q}$
are two distributions of $V$,
\begin{equation}\label{e:classicalprocessing}
\text{TV}\big(\mathbb{P}_V,\mathbb{Q}_V\big) \ge \text{TV}\big(\mathbb{P}_{F(V)}, \mathbb{Q}_{F(V)}\big). 
\end{equation}
Here is a proof of this inequality. Write
$\mathcal W=\{s_1,...,s_c\}$, and for each $i\in\{1,\ldots ,c\}$,
define $\mathcal V_i=\{v:f(v)=s_i\}$. The $\mathcal{V}_{i}$ form a
partition of $\mathcal V$. Then we have 
\begin{eqnarray}
  \text{TV}\big(\mathbb{P}_{F(V)}, \mathbb{Q}_{F(V)}\big)&=& \frac{1}{2} \sum_{i=1}^c \left|\mathbb{P}(V\in \mathcal V_i)-\mathbb{Q}(V \in \mathcal V_i)\right|\notag\\
                                       \ignore{&=& \frac{1}{2} \sum_{i=1}^c \left|\sum_{v\in \mathcal V_i}\mathbb{P}(V=v)-\sum_{v\in \mathcal V_i}\mathbb{Q}(V=v)\right|\notag\\}
                                       &=& \frac{1}{2} \sum_{i=1}^c \left|\sum_{v\in \mathcal V_i}\left[\mathbb{P}(V=v)-\mathbb{Q}(V=v)\right]\right|\notag\\
                                       &\le& \frac{1}{2} \sum_{i=1}^c \sum_{v\in \mathcal V_i}\left|\mathbb{P}(V=v)-\mathbb{Q}(V=v)\right|\notag\\ 
                                       \ignore{&=& \frac{1}{2} \sum_{v\in \mathcal V}\left|\mathbb{P}(V=v)-\mathbb{Q}(V=v)\right|\notag\\}
                                       &=& \text{TV}\big(\mathbb{P}_V,\mathbb{Q}_V\big).
\end{eqnarray} 

\begin{sloppypar}
  We sometimes need to refer to the sequences of RVs associated with the first $i-1$
  trials. To do this we use notation such as $({\bf AB})_{<i}$ for the outcome
  sequence $A_1B_1A_2B_2...A_{i-1}B_{i-1}$, $({\bf XY})_{<i}$ for the
  settings sequence $X_1Y_1...X_{i-1}Y_{i-1}$, and $({\bf ABXY})_{<i}$
  for the joint outcomes and settings sequence
  $A_1B_1X_1Y_1...A_{i-1}B_{i-1}X_{i-1}Y_{i-1}$.  In general we often
  juxtapose RVs to indicate the ``joint'' RV.  From our assumptions
  Eqs.~\ref{e:mtunifsettings} and~\ref{e:mtnosig} and the fact that
  $\text{past}_{i}$ subsumes the trial settings and outcomes from
  trials $1$ through $i-1$, we obtain
\begin{equation}\label{e:indeppast}
  \forall i \in (1,...,n), \quad \mathbb{P}_e\left(X_iY_i|({\bf ABXY})_{<i}\right) = \mathbb{P}_e(X_iY_i) = 1/4,
\end{equation}
and
\begin{eqnarray}\label{e:nosig}
\mathbb{P}_e(A_i|X_iY_i, ({\bf ABXY})_{<i})&=&\mathbb{P}_e(A_i|X_i, ({\bf ABXY})_{<i})\notag \\ 
\mathbb{P}_e(B_i |X_iY_i, ({\bf ABXY})_{<i})&=&\mathbb{P}_e(B_i |Y_i, ({\bf ABXY})_{<i}).
\end{eqnarray}
These are the forms of our assumptions used in the proof
of the Entropy Production Theorem.
\end{sloppypar}

For a generic trial of a two station Bell test, a distribution is
defined to be non-signaling if
\begin{equation}\label{e:nosiggen}
\mathbb{P}(A|XY)=\mathbb{P}(A|X) \quad \text{and} \quad
\mathbb{P}(B |XY)=\mathbb{P}(B |Y).
\end{equation}
Such distributions form a convex polytope and include the
\textit{local realist} (LR) distributions.  Using the conventions
of~\cite{BBP}, these are defined as follows: Let $\lambda$ range over
the set of sixteen four-element vectors of the form
$(a_0,a_1,b_0,b_1)$ with elements in $\{\text{+},0\}$.  Each $\lambda$
induces settings-conditional deterministic
distributions according to
\begin{equation}\label{e:localdet}
\mathbb{P}_\lambda(ab|xy) = \begin{cases}
1, & \text{ if $a=a_x$ and $b=b_y$,}\\
0, & \text{ otherwise.}\\
\end{cases}
\end{equation}
Then a probability distribution $\mathbb{P}$ is LR iff its
conditional probabilities $\mathbb{P}(ab|xy)$ can be written as a convex
combination of the $\mathbb{P}_\lambda(ab|xy)$. That is
\begin{equation}\label{e:local}
\mathbb{P}(ab|xy)=\sum_\lambda q_\lambda \mathbb{P}_\lambda(ab|xy),
\end{equation}
with $q_{\lambda}$ a $\lambda$-indexed set of nonnegative numbers summing to 1. 
This definition agrees with the one given in the main text. 

The eight ``Popescu-Rohrlich (PR)
  boxes'' \cite{PRBOX} are examples of non-signaling distributions
  that are not LR. One of the PR boxes is defined by
\begin{equation}\label{e:PRbox}
\mathbb{P}_{\text{PR}}(ab|xy)=\begin{cases} 1/2 & \text{ if } xy\ne 11 \text{ and }a=b, \text{ or if } xy = 11 \text{ and }a\ne b,\\
0 & \text{ otherwise,}
\end{cases}
\end{equation}
and the other seven are obtained by relabeling settings or outcomes.
We take advantage of the facts that a PR box contains one bit of
randomness conditional on the settings and that the PR boxes
together with the $16$ deterministic LR distributions of \eqref{e:localdet} form the set of
extreme points of the non-signaling polytope~\cite{barrett:2005}.

\subsection{Proof of the Entropy Production Theorem} 
\label{st:ept}

The conditions on $T$ given in the main text are that $T>0$,
$\mathbb{E}(T)_\mathbb{P}\leq 1$ for every LR distribution $\mathbb{P}$ and $\mathbb{E}(T)_\mathbb{Q}\leq 1+m$
for every non-signaling distribution $\mathbb{Q}$, where the bound $1+m$ is
achievable.  In the following, we call these the Bell-function
conditions with bound $m$.  Here we prove the Entropy Production
Theorem in the more general form where the $T_{i}$ can be chosen based
on $({\bf abxy})_{<i}$.  We call $T_{i}$ a past-parametrized family of
Bell functions if for all $({\bf abxy})_{<i}$,
$T_{i}(a_{i}b_{i}x_{i}y_{i},({\bf abxy})_{<i})$ satisfies the
Bell-function conditions with bound $m'\leq m$ when considered as a
function of the results $a_{i}b_{i}x_{i}y_{i}$ from the $i$'th
trial. The theorem and its proof can also be directly applied to the
special case where $T_i$ is the same function for all trials $i$.

\begin{Theorem}\label{t:ept}
  Let $T_{i}$ be a past-parametrized family of Bell functions as
  defined in the previous paragraph.  Then in an experiment of $n$
  trials obeying Eqs.~\ref{e:mtunifsettings} and~\ref{e:mtnosig}, the
  following inequality holds for all $\epsilon_{\mathrm{p}} \in (0,1)$
  and $\Vthresh$ satisfying
    $1\le \Vthresh \le (1+(3/2)m)^{n}\epsilon_{\mathrm{p}}^{-1}$:
  \begin{equation}
    \mathbb{P}_e\left(\mathbb{P}_e({\bf AB}|{\bf XY})> \delta , V\ge \Vthresh \right) \le\epsilon_{\mathrm{p}}
  \label{e:1}
  \end{equation}
  where $\delta =
  [1+(1-\sqrt[n]{\epsilon_{\mathrm{p}}\Vthresh})/2m]^n$ and
  $\mathbb{P}_e$ represents the probability distribution conditioned on the
  event $\{E=e\}$.
\end{Theorem}

We include the constraint
$\Vthresh\leq(1+(3/2)m)^{n}\epsilon_{\text{p}}^{-1}$ for technical
reasons. Higher values of $\Vthresh$ are unreasonably large and result
in pass probabilities that are too low to be relevant. Note that this
bound ensures $\delta\ge 2^{-2n}$, a fact that will be useful in
proving the Protocol Soundness Theorem in (\ref{st:pst}).

\begin{proof}
  Since the condition on $\{E=e\}$ appears uniformly throughout, in
  this proof we omit the subscript on $\mathbb{P}_{e}$ specifying conditioning
  on $\{E=e\}$.

  The strategy of the proof is to first obtain an upper bound on the
  one-trial outcome probabilities from the expectations of Bell
  functions $T$.  This bound can be chained to give a bound on the
  probabilities of the outcome sequence as a monotonically decreasing
  function of the product of the conditional expectations of the
  $T_{i}$. That is, a larger product of expectations yields a smaller
  maximum probability and therefore more extractable randomness.  This
  product cannot be directly observed, so we relate it to the observed
  product $V$ of the $T_{i}$ via the Markov inequality applied to an
  associated positive, mean-$1$ martingale. In the following, we
  suppress the arguments $a_{i}b_{i}x_{i}y_{i}$ and
  $({\bf ABXY})_{<i}$ of $T_{i}$.
    
  The one-trial outcome probabilities are bounded by means of the
  following lemma:

  \begin{Lemma} Let $T$ satisfy the Bell-function conditions with
    bound $m>0$.  For any non-signaling
      distribution $\mathbb{P}$ with equiprobable settings
      (Eqs.~\ref{e:indeppast} and \ref{e:nosiggen}),
    \begin{equation}\label{e:maxprobbound} 
      \max_{abxy}\mathbb{P}(ab|xy)\le 1+
      \frac{1-\mathbb{E}[T(A,B,X,Y)]_\mathbb{P}}{2m}.
    \end{equation}
  \end{Lemma}

  \begin{proof} As $\mathbb{P}$ is a non-signaling distribution with
    equiprobable settings, it can be obtained as a convex combination
    of extremal such distributions. In particular, $\mathbb{P}$ can be
    expressed as such a convex combination containing at most one PR
    box (\cite{bierhorst:2016}, Corollary 2.1), so we can write
    $\mathbb{P}=p\mathbb{Q}+(1-p)\mathbb{Q}'$, where $\mathbb{Q}$ is
    the PR box and $\mathbb{Q}'$ is LR.  By assumption,
    $\mathbb{E}(T)\leq 1$ for all LR distributions with uniform
    settings choices, so $\mathbb{E}(T)_{\mathbb{Q}'}\leq 1$ and
    $\mathbb{E}(T)_{\mathbb{Q}} \leq 1+m$.  Consequently,
    $\mathbb{E}(T)_{\mathbb{P}}=p\mathbb{E}(T)_{\mathbb{Q}}+(1-p)\mathbb{E}(T)_{\mathbb{Q}'}\leq
    p(1+m)+(1-p)=1+pm$, or equivalently, $p\geq
    (\mathbb{E}(T)_{\mathbb{P}}-1)/m$.  The PR box assigns
    $xy$-conditional probability $1/2$ to at least one outcome
    different from $ab$. It follows that the $xy$-conditional
    probability relative to $\mathbb{P}$ of an outcome different from
    $ab$ is at least $p/2$. Therefore, $\mathbb{P}(ab|xy)\le 1-p/2 \le
    1-(\mathbb{E}(T)_{\mathbb{P}}-1)/(2m)$. Since $ab$ and $xy$ are
    arbitrary, this gives the inequality in the lemma.
  \end{proof}

  The inequality in the lemma holds if $T$ has bound $m'\leq m$.  If
  $\mathbb{E}(T)_\mathbb{P}\leq 1$ this is trivial. If $1<\mathbb{E}(T)_\mathbb{P}\leq m'$, the lemma
  holds with $m$ substituted by $m'$, giving a lower upper bound on
  the maximum probability. With this observation, and the fact that by
  assumption, $\mathbb{P}(a_ib_i|(\mathbf{abxy})_{<i}, x_iy_i)$ is
  non-signaling with respect to $a_{i}, b_{i}, x_{i},$ and $y_{i}$, we
  can establish a bound on $\mathbb{P}({\bf ab}|{\bf xy})$ as follows:
  \begin{eqnarray}
    \mathbb{P}({\bf ab}|{\bf xy})
      &=& \prod_{i=1}^n\mathbb{P}(a_ib_i|({\bf ab})_{<i}, {\bf xy})\notag\\
      &=&\prod_{i=1}^n\mathbb{P}(a_ib_i|({\bf abxy})_{<i}, x_iy_i)\notag\\
      &\le& \prod_{i=1}^n\left[1+\frac{1-\mathbb{E}(T_i|({\bf abxy})_{<i})}{2m}\right].\label{e:eptstep2}
  \end{eqnarray}
  Here, the first identity is the chain rule for conditional
  probabilities, and the second follows from \eqref{e:indeppast}.  By
  twice using the fact that the geometric mean of a set of positive
  numbers is always less than or equal to their arithmetic mean, we
  continue from the last line of \eqref{e:eptstep2}:
  \begin{eqnarray} 
    \prod_{i=1}^n\left[1+\frac{1-\mathbb{E}(T_i|({\bf
          abxy})_{<i})}{2m}\right] &=&
    \left(\left\{\prod_{i=1}^n\left[1+\frac{1-\mathbb{E}(T_i|({\bf
              abxy})_{<i})}{2m}\right]\right\}^\frac{1}{n}\right)^n\notag\\
    &\le&\left(\frac{\sum_{i=1}^n\left[1+\frac{1-\mathbb{E}(T_i|({\bf
              abxy})_{<i})}{2m}\right]}{n}\right)^n\notag\\
    &=&\left(1+\frac{1}{2m}-\frac{\sum_{i=1}^n\left[\frac{\mathbb{E}(T_i|({\bf
              abxy})_{<i})}{2m}\right]}{n}\right)^n\notag\\
    &\le&\left(1+\frac{1}{2m}-\left[\prod_{i=1}^n\frac{\mathbb{E}(T_i|({\bf
            abxy})_{<i})}{2m}\right]^\frac{1}{n}\right)^n\notag\\
    &=&\left(1+\frac{1-\left[\prod_{i=1}^n\mathbb{E}(T_i|({\bf
            abxy})_{<i})\right]^{\frac{1}{n}}}{2m}\right)^n. \label{e:geomean}
  \end{eqnarray}

  \begin{sloppy}
    Referring back to the statement of the theorem, we see that
    $\delta$ can be expressed as $f(\epsilon_{\text{p}}\Vthresh)$
    where $f(x)=[1+(1-\sqrt[n]{x})/2m]^n$.  Expressing
    \eqref{e:geomean} in terms of this same function $f$, we see
    that the event $\{\mathbb{P}({\bf AB}|{\bf XY})> \delta\}$ implies the
    event $\left\{f\left(\prod_{i=1}^n\mathbb{E}(T_i|({\bf ABXY})_{<i})\right)
      > \delta\right\}$. The latter event is the same as
    $\left\{\prod_{i=1}^n\mathbb{E}(T_i|({\bf ABXY})_{<i})<f^{-1}(
      \delta)=\epsilon_{\text{p}}\Vthresh\right\}$, since $f^{-1}$ is
    strictly decreasing. Conjoining the event $\{V\geq \Vthresh\}$ to
    both sides of the implication, we have $\{\mathbb{P}({\bf AB}|{\bf XY})>
    \delta,V\geq \Vthresh\}$ implies $\left\{\prod_{i=1}^n\mathbb{E}(T_i|({\bf
        ABXY})_{<i})<\epsilon_{\text{p}}\Vthresh,V\geq
      \Vthresh\right\}$, and so by the monotonicity of probabilities,
    \begin{equation}\label{e:alt1}
      \mathbb{P}\left(\mathbb{P}({\bf AB}|{\bf XY})> \delta,
        V\ge \Vthresh\right) \leq \mathbb{P}\left(\prod_{i=1}^n\mathbb{E}(T_i|(\bfABXY)_{<i})<\epsilon_{\text{p}}\Vthresh , V \geq \Vthresh\right).
    \end{equation}
    The event $\{\Phi\}$ whose probability appears on the left-hand
    side of this equation is the event in the theorem statement whose
    probability we are required to bound.  For any values of the RVs,
    the two inequalities in the event on the right-hand side imply the
    inequality in the event
    $\{\Psi\}=\left\{V/\prod_{i=1}^n\mathbb{E}(T_i|(\bfABXY)_{<i})\ge
      1/\epsilon_{\text{p}}\right\}$. Hence $\mathbb{P}(\Phi)\leq \mathbb{P}(\Psi)$.  It
    remains to show that $\mathbb{P}(\Psi)\leq \epsilon_{\text{p}}$. For this
    purpose we define the sequence $\{W_c\}_{c=1}^{n}$ of RVs by
    \begin{equation}
      W_c = \prod_{i=1}^{c}\frac{T_{i}}{\mathbb{E}(T_i|(\bfABXY)_{<i})},
    \end{equation}
    so that $\{\Psi\}=\{W_{n}\geq 1/\epsilon_{\text{p}}\}$.
    
    By definition, $W_{c}> 0$ and the factors
    $T_{i}/\mathbb{E}(T_{i}|(\bfABXY)_{<i})$ have expectation $1$ conditional
    on the past. Sequences of RVs with these properties are
    referred to as test martingales~\cite{shafer:2009} and satisfy
    that $\mathbb{E}(W_{n})=1$, which can be verified directly by induction:
    \begin{align}
      \mathbb{E}(W_c|({\bf ABXY})_{<c}) &=    \mathbb{E}\left(\prod_{i=1}^{c}\frac{T_{i}}{\mathbb{E}(T_i|(\bfABXY)_{<i})}\middle|({\bf ABXY})_{<c}\right)\notag\\
      &=    \mathbb{E}\left(\left(\prod_{i=1}^{c-1}\frac{ T_{i}}{\mathbb{E}(T_{i}|({\bf ABXY})_{<i})} \right)\frac{1}{\mathbb{E}(T_{c}|({\bf ABXY})_{<c})}T_{c}\middle|({\bf ABXY})_{<c}\right)\notag\\
      &=    \left(\prod_{i=1}^{c-1}\frac{ T_{i}}{\mathbb{E}(T_{i}|({\bf ABXY})_{<i})}\right) \frac{1}{\mathbb{E}(T_{c}|({\bf ABXY})_{<c})}\mathbb{E}\left(T_{c}\middle|({\bf ABXY})_{<c}\right)\notag\\
      &= W_{c-1},\label{e:alt2}
    \end{align}
    where in the second last line, we pulled out factors that are
    functions of the conditioner $(\bfABXY)_{<c}$ by applying the rule
    that if $F$ is a function of $H$, then $\mathbb{E}(FG|H)=F\mathbb{E}(G|H)$.  Taking
    the unconditional expectation of both sides of \eqref{e:alt2} and
    invoking the law of total expectation, we have
    $\mathbb{E}(W_c)=\mathbb{E}(W_{c-1})$, and so inductively, $\mathbb{E}(W_n)=\mathbb{E}(W_1)$.  Since
    $\mathbb{E}(W_{1})=1$, the claim follows. To finish the proof of the
    Entropy Production Theorem, we apply Markov's inequality to obtain
    $\mathbb{P}(W_{n}\geq 1/\epsilon_{\text{p}})\leq\epsilon_{\text{p}}$ and
    consequently $\mathbb{P}(\Phi)\leq\epsilon_{\text{p}}$.
  \end{sloppy}

\end{proof}

Now that we have proved the Entropy Production Theorem for any
past-parametrized family of Bell functions, we can justify the
strategy described in the main text and used in our protocol, where we
set outcomes to $0$ after $\Vthresh$ is exceeded by the running
product. Given the constraints on $T$ used in the protocol, this
strategy is equivalent to setting the remaining Bell functions to
$T_{i}=1$ as allowed by past-parametrization.  Formally, since the
running product $V_{i-1}=\prod_{i=j}^{i-1}T_{j}$ is a function of
$({\bf ABXY})_{<i}$, we can define $T_{i}=T$ conditional on
$\{V_{i-1}<\Vthresh\}$ and $T_{i}=1$ conditional on the complement.

\subsection{Choosing the Bell Function $T$}
\label{st:T} 

The Entropy Production Theorem does not indicate how to find functions
$T$ satisfying the specified conditions. We seek a high typical value
of $V=\prod_{i=1}^{n}T_{i}$, as this permits larger values of
$\Vthresh$ and consequently more extractable randomness at the same
values of $\epsilon_{\text{p}}$ and $m$. Here, we describe a procedure
for constructing a function $T$ that can be expected to perform well
if the trial results are i.i.d.~with known distribution. We estimate
the distribution from an initial portion of the run that we set aside
as training data, and in a stable experiment we expect that the trial
results' statistics are i.i.d.~to a good approximation. Note however
that the optimistic i.i.d.~assumption is only used as a heuristic to
construct $T$; once $T$ is chosen the guarantees of the Entropy
Production Theorem hold regardless of whether the trial results are
actually i.i.d.  

The observed measurement outcome frequencies for training data
generally yield a weakly signaling distribution that does not exactly
satisfy the non-signaling constraints in \eqref{e:nosiggen}, due to
statistical fluctuation. Hence one can obtain an estimated
distribution by determining the maximum likelihood non-signaling
distribution for the observed measurement outcomes frequencies as
described in Ref.~\cite{zhang:2011}.  Let $N(xy)$ be the number of
training trials at setting $xy$ and $f(ab|xy)=N(ab|xy)/N(xy)$ be the
empirical frequencies of outcome $ab$ given setting $xy$. Let
$\mathbb{Q}(a,b,x,y)$ be a candidate for the probability distribution
from which these frequencies were sampled. Then up to an additive term
independent of $\mathbb{Q}$ accounting for the settings probabilities,
the log-likelihood of $f$ given $\mathbb{Q}$ is
$L(\mathbb{Q})=\sum_{a,b,x,y}N(xy)f(ab|xy)\ln(\mathbb{Q}(a,b|x,y))$. We
maximized a variant of this function to find our estimated
distribution $\mathbb{Q}(a,b,x,y)$:
\begin{align}\label{e:convexfindNS}
&\underset{\mathbb{Q}}{\text{Maximize }} \sum_{abxy}f(ab|xy)\ln \mathbb{Q}(a,b,x,y)\\
&\begin{array}{lrcll}
\!\!\text{Subject to }& \mathbb{Q}(x,y)&=&1/4 & \text{for}\quad x, y \in\{0,1\}\\
 & \mathbb{Q}(a|x,y)&=&\mathbb{Q}(a|x) & \text{for} \quad x,y \in \{0,1\}, \quad a \in \{\text{+},0\}\\
 & \mathbb{Q}(b|x,y)&=&\mathbb{Q}(b|y) & \text{for} \quad x,y \in \{0,1\}, \quad b \in \{\text{+},0\}.
\end{array}\notag
\end{align}
The first group of constraints encode our knowledge that all settings
combinations are equiprobable, and the remaining constraints are
the non-signaling constraints. Note that the conditional expressions
in these constraints are equivalently expressed as linear functions of
$\mathbb{Q}(a,b,x,y)$ after using the identities $\mathbb{Q}(x,y)=1/4$.

Once the estimated distribution $\mathbb{Q}$ is obtained, we maximize
the typical values of $V$ by taking advantage of the observation that
the conditions on $T$ imply that $V^{-1}$ is a conservative $p$-value
against local realism \cite{zhang:2011}.  Such $p$-values were studied
in Ref.~\cite{zhang:2011}, which gives a general strategy, the PBR
method, for maximizing $\mathbb{E}(\ln(V))_\mathbb{Q}$. This is useful
because typical values of $V$ are close to
$\exp({n\mathbb{E}(\ln(T))_\mathbb{Q}})$: Since
$\ln(V)=\sum_{i=1}^{n}\ln(T_{i})$ is a sum of i.i.d.\ bounded terms
(given our optimistic assumption), the central limit theorem ensures
that $\ln V$ is approximately normally distributed with mean
$n\mathbb{E}(\ln(T))_\mathbb{Q}$. We therefore perform the following
optimization problem to find $T$:
\begin{align}\label{e:convexfindT}
&\underset{T}{\text{Maximize }} \mathbb{E}(\ln (T))_{\mathbb{Q}}\\
&\begin{array}{lrcll}
\!\!\text{Subject to }& \mathbb{E}(T)_{\mathbb{P}_\lambda} &\leq&1 & \forall \lambda\\
 & T(0,0,x,y) &=&1 & \forall x,y,\\
\end{array}\notag
\end{align}
where $\mathbb{P}_\lambda$ refers to the 16 conditionally deterministic LR
distributions in \eqref{e:localdet}.  This ensures that
$\mathbb{E}(T)_{\mathbb{P}_{LR}}\le 1$ for all LR distributions $\mathbb{P}_{LR}$. The second
constraint is motivated by the fact that in our experiments, an
overwhelming fraction of the trials have no detections for both
stations.  While it is possible that a better $\mathbb{E}(\ln(T))_{\mathbb{Q}}$ can be
obtained without this constraint, we have found that the improvement
is small and likely not statistically significant given the amount of
training data used to determine the results distribution. Since the
objective functions are concave and the constraints are linear, the
optimization problems given in \eqref{e:convexfindNS} and
\eqref{e:convexfindT} are readily solved numerically with standard
tools.

Given the assumption that the trial results are i.i.d., the previous
paragraph shows that the typical values for $V$ are exponential in the
number of trials, $V = e^{-n\mathbb{E}(\ln(T))-o(n)}$.  If the
experiment is successful in showing violation of local realism,
$\mathbb{E}(\ln(T))$ is positive.  Neglecting the contribution from
$o(n)$, with $\Vthresh=e^{n\mathbb{E}(\ln(T))}$, we can bound
$-\ln(\delta)$ as
\begin{eqnarray}
  -\ln(\delta) &=& -n\ln(1+(1-(\epsilon_{\text{p}}e^{n\mathbb{E}(\ln(T))})^{1/n})/(2m))\notag\\
    &=& -n\ln(1+(1-e^{\mathbb{E}(\ln(T))+\ln(\epsilon_{\text{p}})/n})/(2m))\notag\\
     &\geq& -n(1-e^{\mathbb{E}(\ln(T))+\ln(\epsilon_{\text{p}})/n})/(2m)\notag\\
     &=& n(e^{\mathbb{E}(\ln(T))+\ln(\epsilon_{\text{p}})/n}-1)/(2m)\notag\\
     &\geq& (n\mathbb{E}(\ln(T))+\ln(\epsilon_{\text{p}}))/(2m).\label{e:yikai}
\end{eqnarray}
where we used $-\ln(1+x)\geq -x$ and $e^{x}-1\geq x$.  This shows that
asymptotically (with $\epsilon_{\text{p}}$ and $\kappa$ constant) we
get at least $\mathbb{E}(\ln(T))\log_{2}(e)/(2m)=\mathbb{E}(\log_2(T))/(2m)$ bits of
randomness per trial. For the empirical distribution obtained from the
XOR 3 trials used for the protocol according to
\eqref{e:convexfindNS}, we obtain
$\mathbb{E}(\log_{2}(T))/2m=1.19\times 10^{-5}$.  The bound in
Eq.~\ref{e:yikai} shows that we can get an asymptotically positive
number of bits of randomness per trial even with $\epsilon_{\text{p}}$
exponentially small in $n$.

\subsection{The TMPS Algorithm}
\label{st:trevisan}

A strong randomness extractor with parameters
$(\sigma, \epsilon,q,d,t)$ is a function
$\text{Ext}:\{0,1\}^{q}\times \{0,1\}^d \to \{0,1\}^t$ with the
property that for any random string $R$ of length $q$ and min-entropy
at least $\sigma$, and an independent, uniformly distributed seed
string $S$ of length $d$, the distribution of the concatenation
$\text{Ext}(RS)$ with S of length $t+d$ is within TV distance
$\epsilon$ of uniform. There are constructions of extractors that
extract most of the input min-entropy $\sigma$ with few seed bits. For
a review of the achievable asymptotic tradeoffs, see
Ref.~\cite{vadhan:2012}, chapter~6.  For explicit extractors that
perform well if not optimally, we used a version of Trevisan's
construction~\cite{trevisan:2001} implemented by Mauerer, Portmann and
Scholz \cite{mauerer:2012}, which we adapted\footnote{Our adapted
  source code is available at \url{https://github.com/usnistgov/libtrevisan}.} to make it functional in our environment and to
incorporate recent constructions achieving improved parameters
\cite{ma:2012}. We call this construction the TMPS algorithm. For a
fixed choice of $\sigma$, $\epsilon$, and $q$, the TMPS algorithm can
construct a strong randomness extractor for any value $t$ obeying the
following bound:
\begin{equation}\label{e:trev1}
 t+4\log_2 t \le \sigma-6 +  4\log_2(\epsilon).
\end{equation}
Given $t$, the length of the seed satisfies 
\begin{equation} \label{e:trev2} d\le w^2\cdot\max \left\{2, 1+
    \left\lceil[\log_2(t-e)-\log_2(w-e)]/[\log_2e-\log_2(e-1)]\right\rceil\right\}
\end{equation}
where $w$ is the smallest prime larger than
$2\times\lceil\log_2(4qt^2/\epsilon^2)\rceil$. We note that the TMPS
extractors are secure against classical and quantum side information
\cite{mauerer:2012}, and this security is reflected in the parameter
constraints.  Since we do not take direct advantage of this security,
it is in principle possible to improve the parameters in the Protocol
Soundness Theorem.

For the bound on the the number of seed bits given after the Protocol
Soundness Theorem in the main text, we have $q=2n$ and
$\epsilon=\epsilon_{\text{ext}}$.  Since for any $r$, there is a prime
$w\leq 2r$, $w=O(\log(n)+\log(t/\epsilon))=O(\log(nt/\epsilon))$, were
we pulled out exponents from the $\log$, and dropped and arbitrarily
increased the implicit constants in front of each term to match
summands. The coefficient of $w^{2}$ in the bound on $d$ is
$O(\log(t))$, because of the ``minus'' sign in front of the term
containing $w$. Multiplying gives $d=O(\log(t)\log(nt/\epsilon)^{2})$.

\subsection{Proof of the Protocol Soundness Theorem}
\label{st:pst}

The distinction between the stations was needed to establish the
inequality in the Entropy Production Theorem and plays no further role
in this section.  We therefore simplify the notation by abbreviating
${\bf C}={\bf AB}$ and either ${\bf Z}={\bf XY}$ or
${\bf Z}={\bf XY}E$. In the former case $\mathbb{P}(\ldots)$ refers to
probabilities conditional on $\{E=e\}$. Otherwise, $\mathbb{P}(\ldots)$
involves no implicit conditions. The Protocol Soundness Theorem holds
regardless of which definition of ${\bf Z}$ is in force.  We write
$R_{\text{pass}}$ to refer to the RV that takes value $1$ conditional
on the passing event $\{V\ge \Vthresh\}$ and $0$ otherwise. Here is a
restatement of the Protocol Soundness Theorem. The constants
$\epsilon_{\text{p}}$ and $\delta$ appearing below are the same as in
the Entropy Production Theorem.

\begin{Theorem}
  Let $0<\epsilon_{\mathrm{ext}}<1$ and $\mathbb{P}(R_{\mathrm{pass}}=1)\ge
  \kappa>0$.  Suppose $t$ is a positive
  integer satisfying
  \begin{equation}\label{e:mttrev1st}
    t+4\log_2t \le -\log_2 \delta + \log_2 \kappa +5\log_2 \epsilon_{\mathrm{ext}} -11.
  \end{equation}
  Then if $\text{Ext}:\{0,1\}^{2n}\times \{0,1\}^d \to \{0,1\}^t$ is
  obtained by the TMPS algorithm with parameters
  $\sigma=-\log_{2}(2\delta/(\kappa\epsilon_{\mathrm{ext}})$ and
  $\epsilon=\epsilon_{\mathrm{ext}}/2$, and {\bf S} is a random
  bit string of length $d$ independent of the joint distribution of
  ${\bf C},{\bf Z},R_{\mathrm{pass}}$, we have
  \begin{equation}\label{e:pst}
    \mathrm{TV}\big(\mathbb{P}_{{\bf UZS}|R_{\mathrm{pass}}=1}, \mathbb{P}^{\mathrm{unif}}_{{\bf U}}\mathbb{P}^{\mathrm{unif}}_{{\bf S}}\mathbb{P}_{{\bf Z}|R_{\mathrm{pass}}=1}\big) \le
    \epsilon_{\mathrm{fin}}=\epsilon_{\mathrm{p}}/\kappa+\epsilon_{\mathrm{ext}},
  \end{equation}
  where ${\bf U}=\mathrm{Ext}({\bf CS})$ and $\mathbb{P}^{\mathrm{unif}}$ denotes the
  uniform probability distribution.
\end{Theorem}

At this point it is tempting to just apply an extractor to ${\bf AB}$
with parameter $\sigma$ given by the nominal
$\epsilon_{\text{p}}$-smoothed min-entropy
$\sigma=-\log_{2}(\delta)$. However, this does not guarantee the
strong condition \eqref{e:pst}. Specifically, there are three reasons
that \eqref{e:1} of the Entropy Production Theorem does not
immediately support the application of an extractor to ${\bf AB}$. The
first is that as specified, the extractor input should have
min-entropy $-\log_2\max_{\bf ab}\mathbb{P}({\bf AB}={\bf ab})=\sigma$
with no smoothness error. The second is that the settings-conditional
smoothed min-entropies can be substantially smaller than the nominal
one.  The third is that the min-entropy is also affected by the
probability of passing $\kappa$ being less than $1$. Accounting for
these effects requires an analysis of the settings- and
pass-conditional distributions and the extractor parameters specified
in the theorem.

\begin{proof}
  The proof proceeds in two main steps inspired by the corresponding
  arguments in Ref.~\cite{pironio:2013}. In the first we determine a
  probability distribution $\mathbb{P}^*$ that is within $\epsilon_{\text{p}}$
  of $\mathbb{P}$ but satisfies an appropriate bound on the conditional
  probabilities of ${\bf C}$ with probability $1$ rather than
  $1-\epsilon_{\text{p}}$.  The distribution $\mathbb{P}^*$'s marginals agree
  with those of $\mathbb{P}$ on ${\bf ZS}$. The probabilities conditional on
  aborting also agree, and uniformity and independence of ${\bf S}$ is
  preserved.  In the second, we apply a proposition from
  Ref.~\cite{konig:2008} on applying extractors to distributions such
  as $\mathbb{P}^{*}$ whose average maximum conditional probabilities satisfy a
  specified bound.  The proposition enables us to determine the
  extractor parameters that achieve the required final distance
  $\epsilon_{\text{fin}}$ in the theorem.

  The Entropy Production Theorem guarantees that
  $\mathbb{P}(\mathbb{P}({\bf C}|{\bf Z})>\delta,
  R_{\text{pass}}=1)\leq\epsilon_{\text{p}}$. In the case where $E$ is
  included in $\bf{Z}$, this follows by the uniformity in $\{E=e\}$ of
  the theorem's conclusion:
  \begin{eqnarray}
    \mathbb{P}(\mathbb{P}({\bf C}|{\bf Z},E)>\delta, R_{\text{pass}}=1)
    &=& \sum_{e}\mathbb{P}(\mathbb{P}({\bf C}|{\bf Z},E)>\delta, R_{\text{pass}}=1|E=e)\mathbb{P}(E=e)\notag\\
    &=&\sum_{e}\mathbb{P}(\mathbb{P}({\bf C}|{\bf Z},E=e)>\delta, R_{\text{pass}}=1|E=e)\mathbb{P}(E=e)\notag\\
    &\le&\sum_{e}\epsilon_{\text{p}}\mathbb{P}(e)\notag\\
    &=&\epsilon_{\text{p}}.
  \end{eqnarray}

  Using the following construction, one may observe that for any
  random variable $U$ with values in a set of cardinality $K$ and
  $\gamma$ satisfying $1/K\le\gamma$, and any distribution
  $\mathbb{P}'$ of $U$, there exists $\mathbb{P}''$ such that
  $\mathbb{P}''(U)\leq\gamma$ and $\mathbb{P}''$ is within TV distance
  $\mathbb{P}'(\mathbb{P}'(U)>\gamma)$ of $\mathbb{P}'$.  To construct
  $\mathbb{P}''$, for $u$ such that $\mathbb{P}'(u)>\gamma$, set
  $\mathbb{P}''(u)=\gamma$. To compensate for the reduced
  probabilities, increase the values of $\mathbb{P}'$ to obtain those
  of $\mathbb{P}''$ without exceeding $\gamma$ on the set $\{u
  :\mathbb{P}'(u)\le\gamma\}$ so that $\mathbb{P}''$ is a normalized
  probability distribution. This is possible because in constructing
  $\mathbb{P}''$ from $\mathbb{P}'$, the total reduction in
  probability on $\{u:\mathbb{P}'(u)>\gamma\}$ given by
  $r_{-}=\sum_{u:\mathbb{P}'(u)>\gamma}(\mathbb{P}'(u)-\gamma)$ is
  less than the maximum total increase possible given by
  $r_{+}=\sum_{u:\mathbb{P}'(u)\le\gamma}(\gamma-\mathbb{P}'(u))$, as
  a consequence of $\gamma\geq 1/K$.  To see this, compute
  $r_{+}-r_{-} = \sum_{u}(\gamma-\mathbb{P}'(u))\geq
  \sum_{u}(1/K-\mathbb{P}'(u))= 0$.  The distance
  $\text{TV}(\mathbb{P}',\mathbb{P}'')$ is given by
  $\sum_{u:\mathbb{P}'(u)>\gamma}(\mathbb{P}'(u)-\gamma) \le
  \mathbb{P}'(\mathbb{P}'(U)>\gamma)$. We can now construct
  $\mathbb{P}^*$ by defining its conditional distributions on ${\bf
    C}$.  For this, substitute $U\leftarrow {\bf C}$,
  $\mathbb{P}'(U)\leftarrow \mathbb{P}({\bf C}|{\bf
    z},R_{\text{pass}}=1)$, $\gamma\leftarrow
  \delta/\mathbb{P}(R_{\text{pass}}=1|{\bf z})$ and
  $\mathbb{P}''(U)\leftarrow \mathbb{P}^*({\bf C}|{\bf
    z},R_{\text{pass}}=1)$.  The constraint on $\gamma$ is satisfied
  because the upper bound on $\Vthresh$ in the statement of the
  Entropy Production Theorem ensures that $\delta\geq2^{-2n}$.  Each
  conditional distribution satisfies $\mathbb{P}^*({\bf C}|{\bf
    z},R_{\text{pass}}=1)\leq \delta/\mathbb{P}(R_{\text{pass}}=1|{\bf
    z})$, which is equivalent to $\mathbb{P}^*({\bf
    C},R_{\text{pass}}=1|{\bf z})\leq \delta$, and is within TV
  distance $\mathbb{P}\big (\mathbb{P}({\bf C}|{\bf z},
  R_{\text{pass}}=1)>\delta/\mathbb{P}(R_{\text{pass}=1}|{\bf z})\big
  |{\bf z},R_{\text{pass}}=1 \big)$ of $\mathbb{P}_{{\bf C}|{\bf
      z},R_{\text{pass}}=1}$. The joint probability distribution
  $\mathbb{P}^*$ is determined pointwise from the already assigned
  values of $\mathbb{P}^{*}({\bf c}|{\bf z}r_{\text{pass}})$ for
  $r_{\text{pass}}=1$ as
  \begin{equation}
    \mathbb{P}^*({\bf czs}r_{\text{pass}}) =
    \left\{\begin{array}{ll}
        \mathbb{P}^*({\bf c}|{\bf z}r_{\text{pass}})
         \mathbb{P}({\bf zs}r_{\text{pass}}) & \textrm{if $r_{\text{pass}}=1$}\\
        \mathbb{P}({\bf czs}r_{\text{pass}})&\textrm{otherwise}.
      \end{array}\right.
  \end{equation}
  Since the marginal distribution of ${\bf ZS}R_{\text{pass}}$ is
  unchanged, the full TV distance between $\mathbb{P}$ and $\mathbb{P}^*$ is given by
  the average conditional TV distance with respect to
  ${\bf ZS}R_{\text{pass}}$, see \eqref{e:TVsameconditionals}.  Since
  the conditional TV distance is zero when $R_{\text{pass}}=0$ and
  from independence of ${\bf S}$, we obtain
  \begin{eqnarray}
    \text{TV}(\mathbb{P}^*_{{\bf CZS}R_{\text{pass}}},\mathbb{P}_{{\bf CZS}R_{\text{pass}}}) \hspace*{-1in}&&\notag\\
    &=&\sum_{{\bf zs}r_{\text{pass}}}
    \text{TV}\big(\mathbb{P}^*_{{\bf C}|{\bf zs}r_{\text{pass}}},
    \mathbb{P}_{{\bf C}|{\bf zs}r_{\text{pass}}}\big) \mathbb{P}({\bf zs}r_{\text{pass}})
    \notag\\
    &=&\sum_{{\bf zs}r_{\text{pass}}}
    \text{TV}\big(\mathbb{P}^*_{{\bf C}|{\bf zs}r_{\text{pass}}},
    \mathbb{P}_{{\bf C}|{\bf zs}r_{\text{pass}}}\big) \llbracket r_{\text{pass}}=1\rrbracket \mathbb{P}({\bf zs}r_{\text{pass}})
    \notag\\
    &\le& \sum_{{\bf zs}r_{\text{pass}}}
  \mathbb{P}\big(\mathbb{P}({\bf C},R_{\text{pass}}=1|{\bf z})>\delta\big |{\bf z},R_{\text{pass}}=1\big)
   \llbracket r_{\text{pass}}=1\rrbracket
   \mathbb{P}({\bf zs}r_{\text{pass}})\notag\\
    &=& \sum_{{\bf z}r_{\text{pass}}}
  \mathbb{P}\big(\mathbb{P}({\bf C},R_{\text{pass}}=1|{\bf z})>\delta\big |{\bf z},R_{\text{pass}}=1\big)
  \llbracket r_{\text{pass}}=1\rrbracket
   \mathbb{P}({\bf z}r_{\text{pass}})\notag\\
    &=& \sum_{{\bf cz}r_{\text{pass}}}
    \llbracket \mathbb{P}({\bf c}r_{\text{pass}}|{\bf z})>\delta\rrbracket
    \mathbb{P}({\bf c}|{\bf z}r_{\text{pass}})\llbracket r_{\text{pass}}=1\rrbracket
   \mathbb{P}({\bf z}r_{\text{pass}})\notag\\
    &=& \mathbb{P}(\mathbb{P}({\bf C}R_{\text{pass}}|{\bf Z})>\delta , R_{\text{pass}}=1)\notag\\  
    &\leq&\mathbb{P}(\mathbb{P}({\bf C}|{\bf Z})>\delta , R_{\text{pass}}=1 )\notag\\
    &\leq& \epsilon_{\text{p}}.
  \end{eqnarray}
 
  At this point we can also bound the TV distance
  conditional on passing. Since $\mathbb{P}^*(R_{\text{pass}}) =
  \mathbb{P}(R_{\text{pass}})$, we can apply \eqref{e:TVsameconditionals} and
  the above bound on the distance to get
  \begin{eqnarray}
    \epsilon_{\text{p}} &\geq& 
    \text{TV}\big(\mathbb{P}^*_{{\bf CZS}R_{\text{pass}}},\mathbb{P}_{{\bf CZS}R_{\text{pass}}}\big)\notag\\
    &=&
    \sum_{r}  \text{TV}\big(\mathbb{P}^*_{{\bf CZS}|R_{\text{pass}}=r},\mathbb{P}_{{\bf CZS}|R_{\text{pass}}=r}\big)\mathbb{P}(R_{\text{pass}}=r)  \notag\\
    &=&
    \text{TV}\big(\mathbb{P}^*_{{\bf CZS}|R_{\text{pass}}=1},\mathbb{P}_{{\bf CZS}|R_{\text{pass}}=1}\big)\mathbb{P}(R_{\text{pass}}=1).
  \end{eqnarray}
  We conclude that
  \begin{equation}\label{e:step1conclusion}
    \text{TV}\big(\mathbb{P}^*_{{\bf CZS}|R_{\text{pass}}=1},\mathbb{P}_{{\bf CZS}|R_{\text{pass}}=1}\big)\leq \epsilon_{\text{p}}/\mathbb{P}(R_{\text{pass}}=1)\leq\epsilon_{\text{p}}/\kappa.
  \end{equation}

  For the second main step, we need the average ``guessing
  probability'' of ${\bf C}$ given ${\bf Z}$ conditional on
  $\{R_{\text{pass}}=1\}$.  This is given by
  \begin{eqnarray}
    \sum_{{\bf z}} \max_{{\bf c}}(\mathbb{P}^*({\bf c}|{\bf z},R_{\text{pass}}=1))\mathbb{P}({\bf z}|R_{\text{pass}}=1) &\le&
    \sum_{{\bf z}} \frac{\delta}{\mathbb{P}(R_{\text{pass}}=1|{\bf z})}\mathbb{P}({\bf z}|R_{\text{pass}}=1) \notag\\
    &=& \delta \sum_{{\bf z}}\frac{\mathbb{P}({\bf z})}{\mathbb{P}(R_{\text{pass}}=1)} \notag \\
    &\leq& \delta/\kappa
  \end{eqnarray}
  Now we can apply Proposition 1 of Ref.~\cite{konig:2008}. The next
  lemma extracts the conclusion of this proposition in the form we
  need. It is obtained by substituting the variables and expressions
  in the reference as follows: $X\leftarrow {\bf C}$,
  $Y\leftarrow {\bf S}$, $E\leftarrow {\bf Z}$,
  $\mathsf{E}(X,Y) \leftarrow \text{Ext}({\bf CS})$,
  $k\leftarrow -\log_{2}(\delta/\kappa)
  -\log_{2}(2/\epsilon_{\text{ext}})$,
  $\epsilon\leftarrow \epsilon_{\text{ext}}/2$ and the distributions
  are replaced with the corresponding ones that are conditional on
  $\{R_{\text{pass}}=1\}$.  The guessing entropy in the reference is
  the negative logarithm of the guessing probability computed above.

  \begin{Lemma} 
    \begin{sloppy}
      Suppose that $\mathrm{Ext}$ is a strong extractor with
      parameters
      $(-\log_{2}(2\delta/(\kappa\epsilon_{\mathrm{ext}})),\epsilon_{\mathrm{ext}}/2,2n,d,t)$. Write
      ${\bf U}=\mathrm{Ext}({\bf CS})$.  Then we have the following
      bound:
      \begin{equation}
        TV\big (\mathbb{P}^*_{{\bf UZS}|R_{\mathrm{pass}}=1},         \mathbb{P}^{\mathrm{unif}}_{{\bf U}} \mathbb{P}_{{\bf S}} \mathbb{P}^*_{{\bf Z}|R_{\mathrm{pass}}=1}\big) \leq \epsilon_{\mathrm{ext}}. \label{e:lemmapm}
      \end{equation}
    \end{sloppy}
  \end{Lemma}
  
  To apply the lemma, we obtain $\text{Ext}$ by the TMPS algorithm
  with the parameters in the lemma. Expanding the logarithms as
  $\sigma=-\log_{2}(\delta)+\log_{2}(\kappa) +
  \log_{2}(\epsilon_{\text{ext}})-1$ and
  $\log_{2}(\epsilon)=\log_{2}(\epsilon_{\text{ext}})-1$ and
  substituting in Eq.~\ref{e:trev1} gives the requirement
  \begin{equation}
    t+4\log_2 t\le -\log_{2}(\delta) + \log_{2}(\kappa) + 5\log_2(\epsilon_{\text{ext}})-11,
  \end{equation}
  as asserted in the Protocol Soundness Theorem.
  The number of seed bits $d$ is obtained from Eq.~\ref{e:trev2}.

  It remains to determine the overall TV distance conditional on
  passing. Applying \eqref{e:classicalprocessing} with
  $V = {\bf C, Z, S }$ and $F$ defined as
  $F({\bf C, Z, S})=\big(\text{Ext}({\bf C,S}), {\bf Z}, {\bf
    S}\big)$, and applying \eqref{e:step1conclusion}, we have
\begin{equation}\label{e:uzs}
\text{TV}\big(\mathbb{P}^*_{{\bf UZS}|R_{\text{pass}}=1},\mathbb{P}_{{\bf UZS}|R_{\text{pass}}=1}\big)\leq
\text{TV}\big(\mathbb{P}^*_{{\bf CZS}|R_{\text{pass}}=1},\mathbb{P}_{{\bf CZS}|R_{\text{pass}}=1}\big)\leq \epsilon_{\text{p}}/\kappa.
\end{equation}
Then by \eqref{e:TIforTV}, \eqref{e:lemmapm}, and \eqref{e:uzs} we have
\begin{equation}
  \text{TV}\big(\mathbb{P}_{{\bf UZS}|R_{\text{pass}}=1}, \mathbb{P}^{\text{unif}}_{{\bf U}}\mathbb{P}^{\text{unif}}_{{\bf S}}\mathbb{P}^*_{{\bf Z}|R_{\text{pass}}=1}\big)
  \le \epsilon_{\text{ext}} +\epsilon_{\text{p}}/\kappa.
\end{equation}
As $\mathbb{P}^*_{{\bf Z}|R_{\text{pass}}=1}=\mathbb{P}_{{\bf Z}|R_{\text{pass}}=1}$, the statement of the theorem follows.
\end{proof}

\subsection{Protocol Application Details}
\label{st:actual}

Ref.~\cite{shalm:2015} reported data sets from six experiments
obtained under space-like separation. Four additional data sets were
obtained: two in which an extra delay was purposefully implemented so
that space-like separation was not enforced, and two space-like
separated ``blind'' data sets initially not subject to any
analysis. In these data sets, each trial involved multiple pulses of
the source laser, each detected separately. For the analyses
in~\cite{shalm:2015}, the trial outcomes were determined by
aggregating a consecutive sequence of $k$ pulses, with outcome ``$+$''
if there was a detection in any one of these pulses and ``$0$''
otherwise.  For our work, all analyses used $k=7$. This was the
largest group of pulses certified to be space-like separated in the
data runs where space-like separation was enforced.

We first investigated training data from seven of the non-blind sets,
labeled 03\_43, 19\_45, 21\_15, 22\_20, 23\_55 (XOR 1), 00\_25 (XOR
2), and 02\_31 (XOR 3) in the online repository of \cite{shalm:2015},
running the full protocol only on XOR 3.  We then applied the protocol
to the blind data sets.  One of the six experiments reported in
Ref.~\cite{shalm:2015} was temporarily unavailable in the online
repository and was omitted from our investigation.

For each of the non-blind data sets, we determined the Bell function
$T$ from training data consisting of the first $5\times 10^{7}$ trials
as explained in \ref{st:T}.  We chose $5\times 10^7$ trials so that we
could obtain a $T$ using an accurate estimate of the experimental
distribution of measurement outcomes without sacrificing too much data
that could be used for randomness extraction.  Assuming i.i.d.\ trials
and Gaussian statistics according to the central limit theorem, we
then inferred the expected value $n\mu$ and variance $n\sigma^{2}$ of
$\sum_{i=1}^{n}\ln(T_{i})$ on the remaining trials, where $n$ and
  $\mu$ were calculated according to the distribution obtained from
  the optimization problem of \eqref{e:convexfindNS}. Note that under
these assumptions, we treat $\sum_{i=1}^{n}\ln(T_{i})$ as if it were a
sum of independent and bounded RVs. Since
$V = \exp\left(\sum_{i=1}^{n}\ln (T_{i})\right)$ we can then calculate
$\Vthresh$ according to the $0.95$ rule described in the main
text. That is, we set $\Vthresh=e^{n\mu-1.645\sqrt{n}\sigma}$.  Based
on the results, we found that only the last data set, XOR 3, was
anticipated to yield sufficient randomness at low error, consistent
with expectations based on the $p$-values against LR given in
Ref.~\cite{shalm:2015}.  We therefore applied the full protocol only
to XOR 3. We note that while we had prior knowledge of general
statistical features of XOR 3 from the analysis in
Ref.~\cite{shalm:2015}, once the relevant parameters were determined
from the training data, the protocol was run only once on the
remaining data of XOR 3.

We next give details for our analysis of XOR 3, then describe how the
protocol performed on the two blind data sets, and finish with a
discussion of results from tests for non-uniformity in the settings
and for signaling in the settings-conditional outcomes.

Data set XOR 3 consists of 182,161,215 trials. The counts for each
trial result from the first $5\times 10^{7}$ trials are shown in Table
\ref{t:rawcounts}.  The maximum likelihood non-signaling distribution
corresponding to these counts is shown in Table~\ref{t:nosig}. We
determined $T$ from this distribution, the values of $T$ are shown in
Table~\ref{t:PBR} of the main text. The parameter $m$ for $T$ is
$0.0120275$.
  \begin{table}[h]\centering\caption{Result counts for the first $5\times
        10^{7}$ trials of XOR 3.}\label{t:rawcounts}
 \begin{tabular}{ r|c|c|c|c| }
 \multicolumn{1}{r}{}
  &  \multicolumn{1}{c}{$ab=\text{++}$}
 &  \multicolumn{1}{c}{$ab=\text{+}0$}
 &  \multicolumn{1}{c}{$ab=0\text{+}$} 
   &  \multicolumn{1}{c}{$ab=00$}
\\
  \cline{2-5}
   $xy=00$&2483 & 1341 & 1266 & 12496049 \\
 \cline{2-5}
   $xy=01$& 2645 & 1113 & 9095 & 12489487 \\
 \cline{2-5}
   $xy=10$& 2602 & 8295 & 1076 & 12483646 \\ 
 \cline{2-5}
   $xy=11$& 44 & 10869 & 11768 & 12478221 \\
 \cline{2-5}
 \end{tabular}
\end{table} 
 \begin{table}[h]\centering\caption{
Maximum likelihood non-signaling distribution according to
the counts in Table~\ref{t:rawcounts}.}\label{t:nosig}
 \begin{tabular}{ r|c|c|c|c| }
 \multicolumn{1}{r}{}
  &  \multicolumn{1}{c}{$ab=\text{++}$}
 &  \multicolumn{1}{c}{$ab=\text{+}0$}
 &  \multicolumn{1}{c}{$ab=0\text{+}$} 
   &  \multicolumn{1}{c}{$ab=00$}
\\
  \cline{2-5}
   $xy=00$&  0.000049006  & 0.000026663   &  0.000025112  &  0.249899219 \\
 \cline{2-5}
   $xy=01$&  0.000053304  &  0.000022364  &  0.000182341  &  0.249741991 \\
 \cline{2-5}
   $xy=10$&  0.000052435  &  0.000165906  &  0.000021683  &  0.249759976 \\ 
 \cline{2-5}
   $xy=11$& 0.000000876   &  0.000217465  &  0.000234769  &  0.249546890 \\ 
 \cline{2-5}
 \end{tabular}
\end{table} 

The $0.95$ rule for determining $\Vthresh$ given that there are
132,161,215 trials for the protocol yields
$\Vthresh=1.66\times10^6$. While this suggests that we can set
$\kappa=0.95$, the i.i.d.\ assumption cannot be met exactly due to
experimental drift and intermittent faults. This suggests a more
conservative choice for $\kappa$. For $\Vthresh=1.66\times10^6$ and
$\epsilon_{\text{p}}$ and $\epsilon_{\text{ext}}$ in the fixed ratio
$9{:}1$, one can choose $\kappa$ to be as low as
$0.33$ while still meeting the benchmark that \eqref{e:mttrev1st}
allows for an output string of at least $t=256$ bits within
  $\epsilon_{\text{fin}}=\epsilon_{\text{p}}/\kappa+\epsilon_{\text{ext}}=0.001$
  of uniform. We observed that for the six earlier data sets for
which we computed $T$ from training, the corresponding value of
$\Vthresh$ was exceeded five times by the running product of the
$T_{i}$'s on the remaining trials. Hence a choice of $\kappa=0.33$
does not seem unreasonably high. Our choice to fix the ratio $9{:}1$
for $\epsilon_{\text{p}}$ and $\epsilon_{\text{ext}}$ was based on
numerical studies optimizing $t$ in \eqref{e:mttrev1st} with various
fixed values of $\kappa$ and $\epsilon_{\text{fin}}$. This ratio
generally performed well, so we used it for all instances of the
protocol.

Throughout, we did not consider the length $d$ of the seed in making
our choices and determined $d$ from the other parameters according
to \eqref{e:trev2}. For applying the extractor to XOR 3, we
generated 73,947 seed bits. The seed bits were obtained after the
experiment from a random number generator similar to one used to
select the settings \cite{shalm:2015}.  The independence assumption
$\mathbb{P}({\bf S},{\bf C},{\bf Z},R_{\text{pass}})=\mathbb{P}({\bf
  S})\mathbb{P}({\bf C},{\bf Z},R_{\text{pass}})$ required by the
Protocol Soundness Theorem is therefore justified, and this is
consistent with our suggestion in the main text that the seed bits
can be obtained from additional instances of the RVs $X_{i}$ (in
which case the needed independence follows from
\eqref{e:mtunifsettings}).  It took 128 seconds for our computer to
construct the extractor according to the TMPS algorithm and generate
the explicit final output string.

As reported at the online data repository of \cite{shalm:2015}, two
``blind'' data runs were taken at the time of the original experiment
and set aside for analysis at a later date. No analysis was performed
on these data sets until the work reported here.  We refer to these
data sets as ``Blind 1'' and ``Blind 2,'' according to the order in
which they were taken.

Blind 1 consists of 356,464,574 trials. We followed the training
procedure described above on the initial $5\times 10^{7}$ trials.
For the computed value of $\Vthresh$, our
benchmark values of $t$ and $\epsilon_{\text{fin}}$ could not be met,
so we did not formally run the protocol on the remaining trials.  Upon
examining these trials, we found that the final product $V$ of the
$T_{i}$ was $2.023\times 10^7$.  The maximum running product was
$6.884\times 10^{7}$, and the curves for randomness that could have
been extracted from Blind 1 in Fig.~\ref{f:blinds} are based on this
value.

Blind 2 consists of $182,837,253$ trials. The first $2\times 10^7$
trials were set aside for training. This smaller number was chosen
because Blind 2 was taken directly after XOR 3 and expected to be
consistent with this prior data set. The function $T$ obtained from
training performed well when retroactively applied to XOR 3, prompting
us to apply the protocol to the remainder of the data set without
further training. From the training procedure, we obtained
$\Vthresh=8.515\times 10^{8}$, which was predicted to be more than
sufficient for extracting $256$ random bits within $0.001$ of uniform.
We then ran the protocol, but found that the running product of the
$T_{i}$ never exceeded $\Vthresh$. This is therefore an
instance where the protocol aborted.  The failure to exceed
$\Vthresh$ is explained by a dramatic change of the results
statistics after approximately $1.08 \times 10^8$ trials, entering a
non-violating regime that drives $V$ below $10^{-300}$. However, prior
to this change, the statistics demonstrate violation, reaching a
maximum running product of $17528.87$ at trial number 73,057,106.  The
curves for randomness that could have been extracted from Blind 2 in
Fig.~\ref{f:blinds} are based on this value.

Tests for non-uniformity of the settings distribution were performed
and reported in Ref.~\cite{shalm:2015}. The tests showed that a
combination of uncontrolled environmental variables and the
synchronization electronics introduced small biases, particularly in
Alice's settings.  The sizes of the biases were found to be
inconsistent over time, and so their existence or magnitude in any
particular data run cannot be precisely inferred based on results from
other data runs or post-experiment testing. Nevertheless, it is
informative to perform statistical consistency checks for uniformity
of the settings distribution (\eqref{e:mtunifsettings}) within each
run, especially in consideration of the known tendency for small
biases.  Using the tests described in Ref.~\cite{shalm:2015}, we
examined the individual unbiasedness of $X$ and $Y$, and then we
performed a separate test of the independence of $X$ and $Y$. For
these tests we used statistics whose asymptotic distributions would
approach the standard normal with mean $0$ and variance $1$, if the
trials were i.i.d.  For Blind 1, these statistics yielded two-tailed
p-values for unbiasedness of $X$ and $Y$ of $9.9\times 10^{-6}$ and
0.65 respectively, and 0.33 for independence. For Blind 2, these
p-values were 0.04, 0.42, and 0.41, and for XOR 3, they were 0.05,
0.27, and 0.90.

The observed bias for Alice's setting $X$ yielding the notably
significant result in Blind 1 was small: the inferred probability of
setting $1$ is $0.50012 \pm 0.00003$.  A possible concern is that the
settings bias may have been in a direction to systematically bias $V$
upward.  To check for this possibility, we computed
$\mathbb{E}(\ln(T))_\mathbb{F}=5.49\times 10^{-8}$ according to the distribution $\mathbb{F}$
obtained directly from the observed frequencies. We then obtained a
normalized distribution $\mathbb{F}'$ by the transformation
$\mathbb{F}'(a,b,x,y) = (1/4)\mathbb{F}(a,b,x,y)/ \mathbb{F}(x,y)$, so $\mathbb{F}'$ has the same
settings-conditional probabilities as $\mathbb{F}$ but $\mathbb{F}'(x,y)=1/4$ for all
settings $x,y$. We found
$\mathbb{E}(\ln(T))_{\mathbb{F}'}=5.52 \times 10^{-8}>\mathbb{E}(\ln(T))_\mathbb{F}$, supporting the
conclusion that the empirical overall direction of the bias in Blind 1
did not favor larger $V$.  Nonetheless, this motivates potential
future work to strengthen the protocol to allow for a relaxed version
of \eqref{e:mtunifsettings} where the settings distribution is only
assumed to be within some $\epsilon$ of uniform, such as is done in
the statistical arguments of \cite{hensen:2015, shalm:2015,
  giustina:2015, rosenfeld:2016} for falsifying LR. We do not pursue here a precise
quantification of how such a relaxation would decrease the certifiable
entropy in the entropy production theorem.

Tests for signaling were also reported in
Ref.~\cite{shalm:2015}. There are four signaling equalities that can
be independently tested: $\mathbb{P}(A|X=0,Y)=\mathbb{P}(A|X=0)$, $\mathbb{P}(A|X=1,Y)=\mathbb{P}(A|X=1)$, $\mathbb{P}(B|X,Y=0)=\mathbb{P}(B|Y=0)$, and $\mathbb{P}(B|X,Y=1)=\mathbb{P}(B|Y=1)$.  These tests
performed on Blind 1 and Blind 2 showed only expected statistical
variation.  Specifically, the p-values for these statistics were
$0.72$, $0.03$, $0.50$, and $0.84$ for XOR 3, $0.98$, $0.09$, $0.87$
and $0.83$ for Blind 1, and $0.01$, $0.50$, $0.04$, and $0.13$ for
Blind 2.

\ignore{$z$-scores are $-0.36$, $2.18$, $-0.68$, and $0.20$
  for XOR 3, $-0.03$, $1.67$, $0.17$ and $0.22$ for Blind 1, and
  $-2.46$, $-0.67$, $-2.04$, and $-1.53$ for Blind 2.}

\subsection{Performance of Previous Protocols.}
\label{st:previous}

Other protocols in the literature could not be used for our data sets
for the following reasons.  Protocols in Refs.
\cite{miller:2014,colbeck:2011,chung:2014,coudron:2014,miller2:2014}
either apply to different experimental setups or provide only
asymptotic security results as the number of trials $n$ approaches
infinity. The analysis of \cite{thinh:2016} applies to i.i.d.\
scenarios, and the protocol of \cite{vazirani:2012} requires systems
that achieve Bell violations much higher than ours. In contrast, the
protocols of Refs.~\cite{pironio:2010,arnon:2016} would yield
randomness for our results' distribution given a sufficiently large
number of trials. However, they are ineffective for the numbers of
trials in our data sets, which we illustrate with a heuristic
argument. Both protocols are based on the Clauser-Horne-Shimony-Holt
(CHSH) Bell function \cite{CHSH}
\begin{equation}
T^c(a,b,x,y) =\begin{cases}
1 & \text{ if } (x,y)\ne (1,1) \text{ and } a=b\\
1 & \text{ if } (x,y)= (1,1) \text{ and } a\ne b\\
0 & \text{ otherwise. }
\end{cases}
\end{equation}
The statistic $\overline {T^c}=n^{-1}\sum_{i=1}^nT^c_i$
used by these protocols for witnessing accumulated violation satisfies
$\mathbb{E}(\overline {T^c}) \le .75$ under LR, while $\mathbb{E}(\overline
{T^c})=0.750008165$ for the distribution in Table \ref{t:nosig}. The
completely predictable LR theory that only produces ``00'' outcomes
regardless of the settings satisfies $\mathbb{E}(\overline {T^c}) = .75$, but
in an experiment of $n=132,161,215$ trials, this theory can produce a
value of $\overline {T^c}$ exceeding $0.750008165$ with probability
above $0.4$.  Thus, based on this statistic alone, we cannot infer the
presence of any low-error randomness.

The protocol of Ref.~\cite{pironio:2010} (the PM protocol for short, see
\cite{pironio:2013,fehr:2013} for amendments), can be modified to work
with any Bell function, and there are methods for obtaining better
Bell functions \cite{nieto:2014,bancal:2014} or simultaneously using a
suite of Bell functions \cite{nieto:2016}. Here, we demonstrate that
for any choice of Bell function, the method of \cite{pironio:2010} as
refined in \cite{pironio:2013} cannot be expected to effectively
certify randomness from an experiment distributed according to Table
\ref{t:nosig} unless the number of trials exceeds $2.4\times 10^{10}$,
which is much larger than the number of trials in our experiments.

For the most informative comparison to our protocol, we consider the
PM protocol without their additional constraint that the distribution
be induced by a quantum state. To derive a bound on the performance of
the PM protocol, we refer to Theorem 1 of \cite{pironio:2013}.  This
theorem involves a choice of Bell function denoted by $I$ (analogous
to our $T$), a threshold $J_{m}$ (analogous to our
$\Vthresh$) to be exceeded by the Bell estimator $\bar{I}
  = n^{-1}\sum_{i=1}^n I_i$, and a function $f$ that we discuss below.
To be able to extract some randomness, the theorem requires that
\begin{equation}\label{e:pirbound}
nf(J_m-\mu) > 0.
\end{equation}
The parameter $\mu$ is given by $(I_{\text{max}} +
I_{\text{NS}})\sqrt{(2/n)\ln(1/\epsilon)}$ where $I_{\text{max}}$ is
the largest value in the range of the Bell function $I$,
$I_{\text{NS}}\leq I_{\text{max}}$ is the largest possible expected
value of $I$ for non-signaling distributions, and $0<\epsilon\leq 1$
is a free parameter that is added to the TV distance from uniform for
the final output string. Smaller choices of $\epsilon$, which is
analogous to our $\epsilon_{\text{p}}$, are desirable but require
larger $n$ for the constraint \eqref{e:pirbound} to be positive as we
will see below. We also note that \eqref{e:pirbound} is a necessary
but not sufficient condition for extracting randomness; in particular,
we ignore the negative contribution from the parameter $\epsilon'$ of
\cite{pironio:2013} (somewhat analogous to our $\kappa$) as well as
any error introduced in the extraction step.

For \eqref{e:pirbound}, we can without loss of generality consider
only Bell functions for which $0 \le I_L < I_{\text{NS}}\le
I_{\text{max}}$, where $I_L$ is the maximum expectation of $I$ for LR
distributions. Further, because the relevant quantities below are
invariant when the Bell function is rescaled, we can assume $I_{L}=1$.
According to Ref.~\cite{pironio:2013}'s Eq.~8 and the following
paragraph, we can write
$f(x)=-\log_{2}(g(x))$, where $g$ is monotonically decreasing and
concave, and satisfies
\begin{equation}\label{e:maxpir}
\max_{ab}\mathbb{P}(ab|xy)\leq g(\mathbb{E}(I)_{\mathbb{P}})
\end{equation}
for all $xy$ and non-signaling distributions $\mathbb{P}$. (Recall that we are
not using the stronger constraint that $\mathbb{P}$ be induced by a quantum
state.)  According to \eqref{e:maxprobbound} we can define
$g(x)=1+(1-x)/(2(I_{\text{NS}}-1))$. Later we argue that
this definition of $g$ cannot be improved. Substituting into
\eqref{e:pirbound} we get the inequality
\begin{equation}\label{e:pirminbound}
  - n \log_2\left[1+\frac{1-J_m  + (I_{\text{max}} + I_{\text{NS}})\sqrt{\frac{2}{n}\ln{\frac{1}{\epsilon}}}}{2(I_{\text{NS}}-1)}\right] >0.
\end{equation}
Since $2(I_{\text{NS}}-1)$ is
positive, this is equivalent to 
\begin{equation}\label{e:nbound1}
  \sqrt{\frac{2}{n}\ln{\frac{1}{\epsilon}}}<\frac{J_m-1}{I_{\text{max}}+I_{\text{NS}}}.
\end{equation}
Noting that $I_{\text{max}}+I_{\text{NS}}\ge 2I_{\text{NS}}$, this implies
\begin{equation}
  \sqrt{\frac{2}{n}\ln{\frac{1}{\epsilon}}} <\frac{J_m-1}{2I_{\text{NS}}}.
\end{equation}
Thus, the number of trials needed to extract randomness  by
the PM protocol is bounded below according to 
\begin{equation}
  \label{e:pirreq}
  n> 8\frac{\ln(1/\epsilon)I_{\text{NS}}^{2}}{(J_{m}-1)^{2}}.
\end{equation}
For a given anticipated experimental distribution $\mathbb{P}_{\text{ant}}$,
$J_m$ is best chosen to be at most $\mathbb{E}(I)_{\mathbb{P}_{\text{ant}}}$.  Otherwise,
the probability that $\bar I$ exceeds $J_{m}$ is small. However, for
the maximum amount of extractable randomness, $J_{m}$ should be close
to $\mathbb{E}(I)_{\mathbb{P}_{\text{ant}}}$.  Consider the inferred distribution of XOR 3
shown in Table~\ref{t:nosig}.  By following the procedure given in
Section 2 of \cite{bierhorst:2016}, we can write this distribution as
a convex combination of a PR box with weight $p=3.266\times 10^{-5}$
and an LR distribution with weight $1-p$. From this we see that one should choose $J_{m} \leq \mathbb{E}(I)_{\mathbb{P}_{\text{ant}}}
  = pI_{\text{NS}}+(1-p) \leq pI_{\text{NS}}+1$.  Substituting into
\eqref{e:pirreq} and using $\epsilon\le 0.05$ (a rather
high bound on the allowable TV distance from uniform) gives
\begin{equation}
  \label{e:prreq}
  n>8\frac{\ln(1/\epsilon)}{p^{2}}\geq 2.4\times 10^{10},
\end{equation}
which is substantially larger than the number of trials in our data
sets.  

To finish our argument that the PM protocol cannot improve on this
bound under our assumptions, consider the definition of $g$.  If we
could find a function $g'\leq g$ with $g'(x)<g(x)$ for some
$x\in(1,I_{\text{NS}}]$, then $f=-\log_{2}(g')$ might yield a smaller
lower bound on $n$.  Note that for $x\leq 1$, $g'(x)\geq g'(1)$ and
$g'(1)$ must be at least 1 because, referring to \eqref{e:maxpir},
there is a conditionally deterministic LR distribution $\mathbb{P}$ satisfying
$\mathbb{E}(I)_{\mathbb{P}}=1$ and $\max_{ab}\mathbb{P}(ab|xy) =1$. Hence \eqref{e:pirbound} cannot
be satisfied for arguments $x$ of $f(x)=-\log_2(g'(x))$ with
$x\leq 1$. Given $x\in(1,I_{\text{NS}}]$, write
$x=(1-p)+pI_{\text{NS}}$. Let $\mathbb{Q}$ be the PR box achieving
$\mathbb{E}(I)_{\mathbb{Q}}=I_{\text{NS}}$ and $\mathbb{Q}'$ a conditionally deterministic LR
theory achieving $\mathbb{E}(I)_{\mathbb{Q}'}=1$.  Then
$\mathbb{E}(I)_{(1-p)\mathbb{Q}'+p\mathbb{Q}'}=x$. Furthermore, there is a setting $xy$ at which
the LR theory's outcome is inside the support of the PR box's
outcomes. To see this, by symmetry it suffices to consider the PR box
of \eqref{e:PRbox}. Its outcomes are opposite at setting $11$ and
identical at the other three. A deterministic LR theory's outcomes are
opposite at an even number of settings, so either it is opposite at
setting $11$, or it is identical at one of the others. For setting
$xy$, the bound in \eqref{e:maxpir} is achieved for our definition of
$g$. Hence any other valid replacement $g'$ for $g$ must satisfy
$g'(x)\ge g(x)$ for $x\in(1,I_{\text{NS}}]$, and so \eqref{e:pirbound}
with $f(x)=-\log_2(g'(x))$ implies \eqref{e:pirbound} with
$f(x)=-\log_2(g(x))$. Thus the lower bound on $n$ derived above will
apply to $g'$ as well.


\bibliographystyle{unsrtnat}
\bibliography{metabib}

\end{document}